\documentclass[11pt]{article}
\usepackage{amsmath}
\usepackage{amssymb}
\usepackage{amscd}
\usepackage{color}
\usepackage{ulem}
\usepackage{graphicx}

\newtheorem{theorem}{Theorem}[section]
\newtheorem{prop}[theorem]{Proposition}
\newtheorem{cor}[theorem]{Corollary}
\newtheorem{lemma}[theorem]{Lemma}

\newtheorem{remark}[theorem]{Remark}
\newtheorem{remarks}[theorem]{Remarks}
\newtheorem{define}[theorem]{Definition}
\newtheorem{example}[theorem]{Example}

\newtheorem{notation}[theorem]{Notation}

\newcommand{\ord}{\mbox{\rm ord}}

\definecolor{highlight}{rgb}{.5,0,.5}

\newcommand{\ie}{{\it i.e.}}

\newcommand{\id} {{\bf 1}}

\newcommand{\mat}{{\mbox{\rm GL}}}
\newcommand{\gal}{{\mbox{\rm Gal}}}
\newcommand{\sd}{$\sigma\delta$-}

\newcommand{\supp}{{\mbox{\rm supp}}}
\newcommand{\divs}{{\mbox{\rm div}}}
\newcommand{\pdisp}{{\mbox{\rm pdisp}}}
\newcommand{\disp}{\mbox{\rm {disp}}}
\newcommand{\rank}{\mbox{\rm {rank}}}

\def\F{\hbox{\bf F}}

\def\CX{{\mathbb C}}

\def\diag{\hbox{\rm diag}}
\def\gl{\mathfrak{g}l}
\def\frakg{\mathfrak{g}}

\def\R{{\mathcal R}}

\def\bZ {{\mathbb{Z}}}
\def\bZnn{{\mathbb{Z}_{\geq 0}}}
\def\bZp {{\mathbb{Z}_{>0}}}
\def\D{{\mathcal D}}
\def\F{{\mathcal F}}

\def\calL{{\mathcal L}}

\def\calS{{\cal S}}
\def\bfa{{\mathbf a}}
\def\bfb{{\mathbf b}}
\def\bfu{{\mathbf u}}
\def\bfv{{\mathbf v}}
\def\bfy{{\mathbf y}}

\begin{document}

\title{Liouvillian Solutions of Linear Difference-Differential Equations}
\author{Ruyong Feng\footnote{Key Laboratory of Mathematics Mechanization,
Institute of Systems Science, AMSS, CAS, Beijing 100190, China, rfeng2@ncsu.edu.
The author is supported by NKBRPC 2004CB318000 and NSFC10671200. The work
was done during a stay of the first author at
Department of Mathematics, North Carolina State University (NCSU).
The hospitality at NCSU is gratefully acknowledged.} \
\ Michael F. Singer\footnote{North Carolina State University,
Department of Mathematics, Box 8205, Raleigh, North Carolina
27695-8205, USA, singer@math.ncsu.edu. This material is based upon
work supported by the National Science Foundation under Grant No.
CCF-0634123.} \ \ Min Wu\footnote{Shanghai Key Laboratory of
Trustworthy Computing, East China Normal University, Shanghai
200062, China, mwu@sei.ecnu.edu.cn, The author is supported in part
by NSFC-44012140 and NSFC-90718041.}} \maketitle

\begin{abstract}
{For  a field $k$ with an automorphism $\sigma$ and a
derivation $\delta$, 
we introduce the notion of liouvillian solutions of linear
difference-differential systems $\{\sigma(Y) = AY, \delta(Y) = BY\}$
over~$k$  
and characterize the existence of  liouvillian solutions in terms of
the Galois group of the systems.  We will give an algorithm to
decide whether such a system has liouvillian solutions when~$k =
\mathbb{C}(x,t), \sigma(x) = x+1, \delta = \frac{d}{dt}$ and the
size of the system is a prime.}
\end{abstract}

\tableofcontents

\section{Introduction}
\label{introduction} One of the initial and key applications of the
Galois theory of linear differential     equations {is}
to characterize the solvability of such equations in terms  of
liouvillian functions, \ie, functions built up iteratively from
rational functions using exponentiation, integration and algebraic
functions (\cite[Appendix 6]{gray}, \cite[Chapters 1.5 and
4]{putsinger1}). {From the differential Galois theory, a
linear differential equation
$$
L(y) =y^{(n)} + a_{n-1}y^{(n-1)} + \cdots +a_0y = 0,
$$
over~$\CX(x)$ can be solved in these terms if and only if its Galois
group has a solvable subgroup of finite index. This allows us to
conclude that if a linear differential equation $L(y) = 0$ has a
liouvillian solution then it has a solution of the form $e^{\int
f}$} where~$f$ is an algebraic function. This characterization is
the foundation of many algorithms that allow one to decide if an
equation has such solutions and find them if they exist. This theory
and these algorithms  have been developed for {systems of
matrix form~$Y' = AY$
as well as for} more general coefficients.  \\[0.1in]
{For the case of difference equations, the situation is in many ways
not as well developed. A Galois theory of linear difference
equations  is developed in \cite{putsinger2}. Later on
in~\cite{hendrikssinger}, a notion of solving in liouvillian terms
is introduced for linear difference equations of the form $$ L(y) =
y(x+n) + a_{n-1}y(x+n-1) + \cdots +a_0y(x)=0 \quad\mbox{with }\quad
a_i \in \CX(x)$$ and for difference equations in matrix form~$Y(x+1)
= A Y(x)$ where~$A$ is a matrix over~$\CX(x)$.
In~\cite{hendrikssinger}, a characterization of solving in
liouvillian terms is presented in terms of Galois groups and an
algorithm is given to decide whether a linear difference equation
can be solved in liouvillian terms. Remark that
in~\cite{hendrikssinger}, solutions are considered as equivalence
classes of sequences of complex numbers $\bfy = (y(0), y(1),
\ldots)$ where two sequences are equivalent if they agree from some
point onward.}   {A rational function $f$ is identified with the
sequence of the rational values $(f(0), f(1), \ldots)$. The addition
and multiplication on sequences are defined elementwise. Liouvillian
sequences are built up from rational sequences by successively
adjoining solutions of the equations of the form~$\bfy(x+1) =
\bfa(x) \bfy(x)$ or~$ \bfy(x+1) -\bfy(x) = \bfb(x)$ and using
addition, multiplication and {\it interlacing} to define new
sequences (the interlacing of~$\bfu = (u_0, u_1, \ldots ), \bfv =
(v_0, v_1,\ldots) $ is $(u_0,v_0, u_1,v_1, \ldots)$). Similar to the
differential case,
a linear difference equation $L(y) = 0$ can be solved in terms of
liouvillian sequences if and only if its Galois group has a solvable
subgroup of finite index. When this is the case, $L(y) =0$ has a
solution that is the interlacing of {{\it hypergeometric} sequences,
and~\cite{hendrikssinger} shows how to decide if this is
the case.}
The paper \cite{hendrikssinger} also gives examples of equations
which have no hypergeometric solutions but do have solutions that
are interlacings of hypergeometric solutions. Similar results also
apply to  difference
equations in matrix form.\\[0.1in]
We now consider  systems of linear difference-differential equations
of the~form
\begin{equation*}
 Y(x+1,t) = A Y(x,t),\quad \frac{dY(x,t)}{dt} = B Y(x,t)
\end{equation*}
where $A$ and~$B$ are square matrices over~$\CX(x,t)$ and~$A$ is
invertible.
{In~\cite{blw, labahn-li, ziming-etal, min}, some
theories and algorithms have been developed on determining
reducibility and existence of {\it hyperexponential solutions} of
such systems}.  However, {as in the pure difference case,
there are systems which have no
hyperexponential 
solutions but have solutions that are interlacings of
hyperexponential 
solutions.    In this paper, we shall use a Galois theory that
appears as a special case of the Galois theory developed
in~\cite[Appendix]{hardouin-singer} to characterize {\it liouvillian
solutions of linear difference-differential systems}, and then
devise an algorithm to determine whether such a system has
liouvillian solutions when the order of the system is prime.}\\[0.1in]
Throughout the paper, we use~$(\cdot)^T$ to denote the transpose of
a vector or matrix and~${\rm det}(\cdot)$ to denote the determinant
of a square matrix. The symbols $\bZnn$ and~$\bZp$ represent the set
of  nonnegative integers and the set of positive integers,
respectively. Denote by~$\id$ the identity map on the sets in
discussion. For a field $k$, denote by~$\mathfrak{g}l_n(k)$ the set
of all~$n\times n$ matrices over~$k$ and by~$\mat_n(k)$ the set of
all~$n\times n$ invertible matrices over $k$. All
difference-differential systems of the form~$\{\sigma(Y)=AY,
\,\delta(Y)=BY\}$  with~$A\in \mat_n(k)$ and~$B\in
\mathfrak{g}l_n(k)$  that are in discussion in the paper are assumed to be integrable. \\[0.1in]
The paper is organized as follows.  In Section~\ref{galoistheory},
we will {first review some Galois theoretic results
in~\cite{hardouin-singer}} and
then show that
the Galois group of a linear difference-differential system is
solvable by finite group if and only if a certain associated system
has solutions in a tower built up using generalizations of
liouvillian extensions.
In Section~\ref{sec4}, we show that {\it irreducible} systems with
liouvillian solutions must be equivalent to systems of particular
form (Theorem~\ref{thm4}) and refine this result for systems of
prime order. We propose an algorithm for deciding if linear
difference-differential systems of prime order have liouvillian
solutions. At last, some examples are illustrated.\\[0.1in]
 We would
like to thank Reinhart Shaefke for supplying a simple proof of
Lemma~\ref{lem62}.
\section{Galois Theory}\label{galoistheory}
\subsection{Picard-Vessiot extensions and Galois groups}\label{sec1}
{In~\cite{hardouin-singer}, a general Galois theory  is
presented for linear integrable systems of difference-differential
equations involving parameters. When there exists no parameters
this theory yields a Galois theory of difference-differential
systems as above.}
{Let us  recall some notation and results in
\cite{hardouin-singer}.}\\[0.1in]
{A {\it \sd-ring}~$R$ is a commutative ring with unit
endowed with an automorphism~$\sigma$ and a derivation~$\delta$
satisfying $\sigma\delta=\delta\sigma$. $R$ is called a {\it
\sd-field} when~$R$ is a field. An element~$c$ of~$R$ is called a
{\it constant} if $\sigma(c)=c$ and $\delta(c)=0$, \ie, it is a
constant with respect both $\sigma$ and $\delta$. The set of
constants of~$R$, denoted by~$R^{\sigma\delta}$, is a subring, and
it is a subfield if $R$ is a field.}
\\[0.1in]
In this section, unless specified otherwise, we always let $k$ be a
\sd-field of characteristic zero and with an algebraically closed
field of
constants.\\[0.1in]
Consider {a system of the form}
\begin{equation}\label{mixed-eqn}
\sigma(Y)=AY, \quad \delta (Y)=BY
\end{equation}
where $A \in \mat_n(k)$, $B \in \mathfrak{g}l_n(k)$ and~$Y $ is a
vector of unknowns of size~$n$. The integer~$n$ is {called} the
order of the system~\eqref{mixed-eqn}. A \sd ring $R$ is called a
{\it \sd Picard-Vessiot extension}, {or a \sd PV extension for
short, of $k$ for the system~\eqref{mixed-eqn}} if it satisfies the
following conditions
\begin{itemize}
\item [$(i)$] $R$ is a simple \sd ring;
\item [$(ii)$] there exists  $Z \in \mat_n(R)$ such that
$\sigma(Z)=AZ$ and $\delta(Z)=BZ$;
\item [$(iii)$] $R=k[Z, \frac{1}{\det(Z)}]$, that is, $R$ is
generated by entries of~$Z$ and the inverse of the determinant of
$Z$.
\end{itemize}
{Note that if the system (\ref{mixed-eqn}) has a \sd PV
extension,  the commutativity of $\sigma$ and $\delta$
implies
\begin{equation*}
\sigma(B) = \delta(A) A^{-1} + ABA^{-1},
\end{equation*}
which is called the {\it integrability conditions} for the
system~\eqref{mixed-eqn}.} Conversely, if the system
(\ref{mixed-eqn}) satisfies the above integrability conditions
and  the constants of~$k$ are algebraically closed, it is shown in
\cite{blw} and \cite[Appendix]{hardouin-singer} that~\sd PV
extensions for~\eqref{mixed-eqn} exist and are unique up to \sd-$k$ isomorphisms.\\[0.1in]
The following notation will be used throughout the paper.
\begin{notation}
Let $A$ be a square matrix over a \sd ring. For a positive
integer~$m$, denote~$ A_m=\sigma^{m-1}(A)\cdots\sigma(A)A.$ For a
linear algebraic group~$G$,~$G^0$ represents the identity component
of $G$.
\end{notation}
\begin{lemma} \label{lem1}
{\rm [Lemma 6.8 in \cite{hardouin-singer}]}  Let $k$ be a \sd field
and $R$ a simple \sd ring,  finitely generated over $k$ as a \sd
ring. Then there are idempotents~$e_0,\dots,e_{s-1}$ {in
$R$} such that
\begin{itemize}
\item [$(i)$] $R=e_0R\oplus \cdots \oplus e_{s-1}R$;
\item [$(ii)$] $\sigma$ permutes the set $\{e_0R,\cdots,e_{s-1}R\}$.
Moreover, $\sigma^s$ leaves each $e_iR$ invariant;
\item [$(iii)$] each $e_iR$ is a domain and  a simple
      $\sigma^s\delta$-ring.
  \end{itemize}
\end{lemma}
{The following lemma is an analogue to Lemma 1.26 in
\cite{putsinger1}.}
\begin{lemma} \label{lem11}
Let $k$ be  \sd field, $R$ be a \sd-PV extension for the system
$$\sigma(Y) = AY, \quad \delta{Y} = BY \quad\mbox{with~$A \in \mat_n(k)$
and~$B \in \gl_n(k)$} $$ and $e_0, e_1, \dots, e_{s-1}$ be as in
Lemma \ref{lem1}. Then each~$e_iR$ is a $\sigma^s\delta$-PV
extension of $k$ for the system $
\{\sigma^s(Y)=A_s Y, \delta(Y)=BY\}$.
\end{lemma}
\begin{proof}
{Let~$R$ be a \sd PV extension for~$\{ \sigma(Y){=}AY,
\delta(Y){=}BY\}$ and~$F$ be a fundamental matrix over~$R$ for the
system. By Lemma~\ref{lem1}, each~$e_iR$ is a
simple~$\sigma^s\delta$-ring. Clearly,~$e_iF$ are the solutions
of~$\{\sigma^s(Y){=}A_sY, \delta(Y){=}BY\}$
since~$\sigma^s(e_i)=e_i$ and~$\delta(e_i)=0$  for $i=0, \dots,
s-1$. Assume that $e_i\det(F)=0$ for some $i$. Then
$$
\sigma^j(e_i\det(F))=e_{i+j\, \mbox{ mod } s}\det(A_j)\det(F)=0
$$
and thus $e_{i+j \mod s}\det(F)=0$ for $j=1,\cdots,s-1$. Therefore
$$
\det(F)=(e_0+\cdots+e_{s-1})\det(F)=0,
$$
a contradiction. So $e_iF$ is a fundamental matrix
for~$\{\sigma^s(Y){=}A_sY, \delta(Y){=}BY\}$ for each~$i$.
Moreover,~$e_iR=k[e_iF, \frac{1}{e_i\det(F)}]$ for each~$i$. This
completes the proof.}
\end{proof}

\begin{cor} \label{cor11}
Let $d \geq 1$ be a divisor of $s$. Suppose that there exist
idempotents~$e_0, \dots, e_{s-1}$ in $R$  such that~$R=e_0R\oplus
e_1R\oplus \cdots \oplus e_{s-1}R$.  Then  for~$i=0,\dots, s-1$, the
subring $\oplus_{j=0}^{\frac{s}{d}-1} e_{i+jd} R$ of $R$ is a
$\sigma^d\delta$-PV extension of~$k$ for the system
\begin{equation*}
\sigma^d(Y)=A_dY,\quad \delta(Y)=BY.
\end{equation*}
Here we use a cyclic notation for the indices $\{0,\cdots,s-1\}$.
\end{cor}
\begin{proof}
 The proof is similar to that of Lemma \ref{lem11}.
\end{proof}

\begin{define}
{Let  $A \in \mat_n(k)$,  $B \in \mathfrak{g}l_n(k)$ and
$Y$ be a vector of unknowns of size $n$. Let $R$ be a \sd PV
extension for 
$\{\sigma(Y)=AY, \,  \delta (Y)=BY\}$. The
group consisting of all \sd $k$-automorphisms of $R$
is called the {\it \sd Galois group} for the system
 and denoted $\gal(R/k)$.}
\end{define}
Denote by~$\gal(e_0R/k)$ the $\sigma^s\delta$-Galois group
for~$\{\sigma^s(Y)=AY, \delta(Y)=BY\}$.  Without loss of generality,
we assume that $\sigma(e_i)=e_{i+1 \,\mbox{ mod }s}$. {Construct a map $\Gamma$ from~$\gal(e_0R/k)$ to $ \gal(R/k) $ as
follows. Let $ \varphi \in \gal(e_0R/k)$.  For any
$r=r_0+r_1+\cdots+r_{s-1}\in
R$ with~$r_j \in e_jR$ for~$j=0,\dots, s-1$, define} 
$$
\Gamma(\varphi)(r)=\sum_{j=0}^{s-1}\sigma^j\varphi\sigma^{-j}(r_j).
$$
Let $\phi \in \gal(R/k)$. {Clearly, $\phi$ permutes the $e_i$'s by
the proof of Lemma 6.8 in \cite{hardouin-singer}. Define a map
$\Delta: \gal(R/k) \rightarrow \mathbb{Z}/s\mathbb{Z}$ to be
$\Delta(\phi)=i$ if $\phi(e_0)=e_i$.} We then have the following

\begin{lemma}\label{lem4}
{Let  $R$ be a \sd PV
extension for 
$\{\sigma(Y)=AY, \,  \delta (Y)=BY\}$ where~$A \in \mat_n(k)$ and~$B
\in \mathfrak{g}l_n(k)$. Let $\Gamma$ and $\Delta$ be stated above.
Then the map~$\Gamma$ is well-defined, \ie, $\varphi \in
\gal(e_0R/k)$ implies $\Gamma(\varphi)\in \gal(R/k)$. Moreover, the
sequence of group homomorphisms}
\begin{equation*}\label{seq}
   0 \longrightarrow \gal(e_0R/k) \overset{\Gamma}{\longrightarrow} \gal(R/k)
   \overset{\Delta}{\longrightarrow}\mathbb{Z}/s\mathbb{Z}\longrightarrow 0
\end{equation*}
is exact.
\end{lemma}
\begin{proof}
The proof is similar to that of Corollary 1.17 in \cite{putsinger1}.
\end{proof}

\begin{lemma}\label{lem2}
{Suppose that $k$ has no proper algebraic \sd field extension
and~$R$ is a \sd PV extension for $\{\sigma(Y)=AY, \,  \delta
(Y)=BY\}$ where~$A \in \mat_n(k)$ and~$B \in \mathfrak{g}l_n(k)$}.
Then
 $\gal(R/k)^0=\gal(e_0R/k)$.
\end{lemma}
\begin{proof}
Let $\hat{k}$ be the algebraic closure of $k$ in the quotient field
of $e_0R$. Then $\gal(e_0R/k)^0=\gal(e_0R/\hat{k})$. Since $k$ has
no proper algebraic \sd field extension, we have $\hat{k}=k$ and
therefore~$\gal(e_0R/k)=\gal(e_0R/k)^0$.  {From Lemma
\ref{lem4}, it follows that~$\gal(e_0R/k)$ is a closed subgroup of
$\gal(R/k)$ of finite index. The proposition  in
\cite[p.53]{humphreys} then implies the lemma.}
\end{proof}
From \cite{hardouin-singer}, we know that a \sd PV extension $R$
over $k$ is the coordinate ring of a~$\gal(R/k)$-torsor over $k$.
{Let~$E$ be an algebraically closed differential field with a
derivation $\delta$. Clearly, $E(x)$ becomes a \sd field endowed
with the extended derivation~$\delta$ such that $\delta(x) = 0$ and
with an automorphism~$\sigma$ on~$E(x)$ given by~$\sigma|_E = \id$
and~$\sigma(x) = x+1$. For such a field~$E(x)$, we will get an
analogue of Proposition 1.20 in \cite{putsinger1}.} Before we state
the result, let us look at the following

\begin{lemma}\label{lem21}
Let $k$ be a differential field with a derivation~$\delta$,  $S$  a
differential ring extension of~$k$ and $I$ a differential radical
ideal of $S$. Suppose that $S$ is Noetherian as an algebraic ring
and that $I$ has the minimal prime ideal decomposition $\cap_{i=1}^t
P_i$ as an algebraic ideal. Then $P_i$ are differential ideals for
$i=1,\dots, t$.
\end{lemma}
\begin{proof}
Let $f_1 \in P_1$ and select $f_i \in P_i\setminus P_1$ for
$i=2,\cdots,t$. Then~$f=f_1f_2\cdots f_t \in I$. {By
taking a derivation on both sides, we have
$$
\delta(f)=\delta(f_1)f_2\cdots f_t+f_1\delta(f_2)\cdots
f_t+\cdots+f_1f_2\cdots \delta(f_t) \in I,
$$ which implies
$\delta(f_1)f_2\cdots f_t \in P_1$. So $\delta(f_1) \in P_1$ and
$P_1$ is a differential ideal. The proofs for other $P_i$'s are
similar.}
\end{proof}

\noindent {Now let $S$ be a finitely generated \sd ring over $k$ and
$I$ a radical \sd ideal of~$S$. Suppose that $S$ is Noetherian as an
algebraic ring and $I = \cap_{i=1}^s P_i$ is the minimal prime
decomposition of $I$ as an algebraic ideal. Since $S$ is Noetherian,
we have $\sigma(I)=I$,  which implies that $\sigma$ permutes the
$P_i$'s. From Lemma~\ref{lem21}, each $P_i$ is a differential ideal.
Therefore if~$\{P_i\}_{i \in J}$ with~$J$ a subset of~$\{1,
\dots,s\}$  is left invariant under the action of $\sigma$, then
$\cap_{i\in J}P_i$ is a~\sd ideal. We then have the following
result.  We will use the following notation: if $V$ is a variety
defiend over a ring $k_0$ and $k_1$ is a ring containing~$k_0$, we
denote by $V(k_1)$ the points of $V$ with coordinates in $k_1$.

\begin{prop}\label{lem3}
{Let $\tilde{k} = E(x)$ be as in the paragraph preceding
Lemma~\ref{lem21}, $R$  a~\sd PV extension of $\tilde{k}$ for the
system $\{\sigma(Y)=AY, \, \delta(Y)=BY\}$ where~$A \in
\mat_n(\tilde{k})$ and~$B \in \mathfrak{g}l_n(\tilde{k})$,
and~$G=\gal(R/\tilde{k})$. Then the corresponding~$G$-torsor~$Z$ has
a point which is rational over $\tilde{k}$ and  $Z(\tilde{k})$
and~$G(\tilde{k})$
  are
isomorphic. Moreover, $G/G^0$ is cyclic.}
\end{prop}
\begin{proof}
{The notation and  proof will follow that of Proposition 1.20 in
\cite{putsinger1}}. Let $Z_0,\cdots,Z_{t-1}$ be the
$\tilde{k}$-components of $Z$. By Lemma \ref{lem21}, the defining
ideals $P_i$ of $Z_i$ are differential ideals. As in the proof of
Proposition 1.20 in~\cite{putsinger1}, there exists~$B\in
Z_0(\tilde{k})$ such that~$Z_0=BG^0_{\tilde{k}}$ and
$Z=BG_{\tilde{k}}$ where~$G_{\tilde{k}}$ denotes the variety $G$
over $\tilde{k}$. Since $Z(\tilde{k})$ is $\tau$-invariant, we have
$$
BG_{\tilde{k}}=\tau(BG_{\tilde{k}})=A^{-1}\sigma(B)G_{\tilde{k}},
$$
which implies  $B^{-1}A^{-1}\sigma(B) \in G(\tilde{k})$.
There exists $N \in G(\tilde{k}^{\sigma\delta})$ such that
$$
B^{-1}A^{-1}\sigma(B) \in G^0(\tilde{k})N.
$$
Let $H$ be the group generated by $G^0$ and $N$. One sees
that~$\tau(BH_{\tilde{k}})=BH_{\tilde{k}}$ and therefore the
defining ideal $\tilde{I}$ of $BH_{\tilde{k}}$
is~$\sigma$-invariant. Since the set~$BH_{\tilde{k}}$ is the union
of some of the $Z_i$, $\tilde{I}$ is of the form~$\cap _{i \in
J}P_i$ with  $J$ a subset of~$\{0,1,\cdots,t-1\}.$ Hence $\tilde{I}$
is a \sd ideal because each $P_i$ is a differential ideal. Since the
defining ideal $I$ of $Z$ is a maximal~\sd ideal, it follows
that~$\tilde{I}=I$ and so $H=G$.
\end{proof}
 From~$B^{-1}A^{-1}\sigma(B) \in G^0(\tilde{k})N$ in the proof of
Proposition \ref{lem3}, we conclude that~$N$ is a generator of the
cyclic group $G/G^0$.

\begin{define}
{Let $k$ be a \sd field and $R$ be a \sd PV extension of $k$ for the
system
$$
\sigma(Y)=AY,
\quad \delta(Y)= BY
$$ with $A\in \mat_n(k)$ and $B\in
\mathfrak{gl}_n(k)$, and $V \subset R^n$ be the solution space of the system.
The system is said to be
{\it irreducible} over~$k$ if~$V$ has no
nontrivial~$\gal(R/k)$-invariant subspaces.}
\end{define}
In a manner similar to the purely differential case
\cite[p.~56]{putsinger2}, one can show that a
system~$\{\sigma(Y)=AY,\, \delta(Y)=BY\}$ with $A\in \mat_n(k)$ and
$B\in \mathfrak{gl}_n(k)$ is reducible  over $k$} if an only if
there exists $U \in \mat_n(k)$ such that { a change of variables $Z
= UY$ yields} $\sigma(Z)= \tilde{A}Z, \delta(Z) = \tilde{B}Z$ with
\[
\tilde{A} = \left(\begin{array}{cc} A_1 & 0\\ A_2 & A_3\end{array}
\right), \quad\quad   \tilde{B} = \left(\begin{array}{cc} B_1 & 0\\
B_2 & B_3\end{array} \right).
\]
Note that ~$\tilde A$ and~$\tilde B$ are again~$n\times n$
matrices over~$k$.\\[.1in]
We end this section with a concrete way of realizing a
\sd Picard-Vessiot extensions for linear difference-differential
systems over fields of particular form.
We proceed in a manner similar to that of \cite[Chapter 3]{putsinger2}.\\[0.1in]
Let $K$ be a differential field with a derivation $\delta$. Denote
by~$S_K$ the set of all sequences of the form ${\bf a} = (a_0, a_1,
\ldots )$ with ${\bf a}(i) = a_i \in K$. {Define an
equivalence relation on~$S_K$ as follows:
any two sequences ${\bf a}$ and ${\bf b}$ are {\it equivalent} if
there exists $N\in \bZp$ such that ${\bf a}(n) = {\bf b}(n)$ for all
$n>N$. Denote by~${\cal S}_K$ the set of equivalence classes
of~$S_K$ modulo the equivalence relation. One sees that~${\cal S}_K$
forms a differential ring with the addition, multiplication and a
derivation $\delta$ defined on~${\cal S}_K$ coordinatewise.}
Clearly, the map~$\sigma$ given by $\sigma((a_0, a_1, \ldots )) =
(a_1, a_2, \ldots )$ is  an automorphism of ${\cal S}_K$ that
commutes with the derivation~$\delta$.  In addition, any
element~$e\in K$ is
identified with $(e,e,\ldots)$. So we can regard $K$ as a (differential) subfield of ${\cal S}_K$.  \\[0.1in]
{Let~$E$ be an algebraically closed differential field with a
derivation $\delta$. Construct an automorphism~$\sigma$ on~$E(x)$
given by~$\sigma|_E =\id$ and~$\sigma(x) = x+1$ and extend~$\delta$
to be a derivation~$\delta$ on~$E(x)$ such that $\delta(x) = 0$.
Assume that~$K$ is a differential field extension of~$E$ with an
extended derivation~$\delta$.
The map~$E(x)\rightarrow S_K$ given by $f \mapsto (0, \ldots , 0,
f(N), f(N+1), \ldots )$, where $N$ is a non-negative integer such
that $f$ has no poles at integers $\geq N$, induces} a \sd embedding
of $E(x)$ into ${\cal S}_K$.
 Consequently, we may identify any matrix~$M$ over~$E(x)$ with a sequence of matrices~$(0, \ldots, 0, M(N),M(N+1), \cdots)$ where~$M(i)$ means the evaluation of the entries of~$M$
at~$x=i$. So we have the following
\begin{prop}\label{PVseq}
Let~$E\subset K$,~$E(x)$,~${\cal S}_K$  be  as above and
let~$\tilde{k}=E(x)$. Assume that~$E$ and~$K$ have the same
algebraically closed field of constants as differential fields. Let
$R$ be a \sd PV extension for the system
$$
\sigma(Y) = AY, \quad \delta(Y) = BY
$$
where $A \in \mat_n(\tilde{k}) $ and~$ B \in \gl_n(\tilde{k})$.  Let
$N\in \bZp$ such that $A(m)$ and~$B(m)$ are defined and $\det(A(m))
\neq 0$ for all $m \geq N$ and assume that~$\delta(Y) = B(N)Y$ has a
fundamental matrix $\overline{Z} \in \mat_n(K)$. Then there exists a
\sd$\tilde{k}$-monomorphism of $R$ into $\calS_K$. Moreover, {the
entries of any solution of~$\{\sigma(Y)=AY, \, \delta(Y)=BY\}$
 in $\calS_K^n$} lies in the image of $R$ in ${\cal S}_K$.
\end{prop}
\begin{proof} Let $R= \tilde{k}[Y,\frac{1}{\det Y}]/I $ be the \sd PV ring extension
for the system~$\{\sigma(Y)=AY,  \delta(Y)=BY\}$ and let $G$ be its
 Galois group. From Proposition~\ref{lem3}, the
corresponding torsor has a point  $P$ with coordinates
in~$\tilde{k}$. This implies that if we introduce a new matrix of
variables~$X$ and let $Y = PX$, then $R = \tilde{k}[X,\frac{1}{\det
X}]/J$ where $J$ is the defining ideal of $G$. Furthermore,
$\sigma(X) = \tilde{A} X$ and $\delta(X) = \tilde{B}X$
where~$\tilde{A} = \sigma(P)^{-1}A P \in G(\tilde{k})$
and~$\tilde{B} = P^{-1}BP-P^{-1}\delta(P) \in \frakg(\tilde{k})$
where $\frakg$ is the lie
algebra of~$G$.\\[0.1in]
{Define recursively a sequence of matrices~$Z_m \in
\mat_n(K)$ for $m \geq N$:
$$ Z_N = \overline{Z}\quad\mbox{and}\quad  Z_{m+1}=A(m)Z_m \quad\mbox{for any $m \geq
N$}.
$$
The integrability condition on $A$ and $B$ implies that~$Z_m$
satisfies~$\delta(Y) = B(m)Y$ for any~$m\ge N$ and so $Z = ( \ldots
, Z_N, Z_{N+1}, \ldots)$ is a fundamental matrix in
$\mat_n(\calS_K)$ of $\{\sigma(Y)=AY, \, \delta(Y)=BY\}$.}\\[0.1in]
{Remark that the key to proving the proposition is to show that $Z$
generates a \sd Picard-Vessiot extension. Unfortunately, we do not
see a direct way to show this and our proof is a little circuitous.}
Clearly,~$U := P^{-1}Z$ satifies that~$\sigma(U) = \tilde{A} U$
and~$\delta(U) = \tilde{B}U$. Then~$\delta(U(N')) = \tilde{B}(N')
U(N')$  for a sufficiently large $N'$ and therefore $K$ contains a
(differential) Picard-Vessiot extension of $E$  for the
equation~${\delta(U)=\tilde B U }.$ Since $\tilde{B}(N') \in
\frakg(\tilde{k})$, Proposition~1.31 (and its proof) in
\cite{putsinger2} together with the uniqueness of Picard-Vessiot
extensions imply that there exists~$\overline{V} \in G(K)$ such
that~$\delta(\overline{V}) = \tilde{B}(N')\overline{V}$. Define $V
\in \calS_K$ by $V(N') = \overline{V}$ and  $V(m+1) = \tilde{A}(m)
V(m)$ for~$m \geq N'$. Then~$\sigma(V) = \tilde{A} V$ and $\delta(V)
=  \tilde{B}V$ and  $V \in G(\calS_K)$. This implies that the map
from $R= \tilde{k}[X,\frac{1}{\det X}]/J$ to $\calS_K$ given by~$X
\mapsto V$ {is a~\sd $\tilde{k}$-homomorphism}. Since $I$ is a
maximal \sd ideal, this map must be injective, and so is the desired
embedding from $R$ into $\calS_K$.
 \\[0.1in]
{Let $W \in \calS_K^n$ be a solution of~$\{\sigma(Y)=AY, \,
\delta(Y)=BY\}$. For a sufficiently large $M$, $W(M)$ is defined and
is a solution of } $\delta(Y) = B(M)Y$ and~$W(m+1) = A(m) W(m)$ for
$m \geq M$.  Therefore~$W(M) = V(M)D$ for some constant vector $D$
and thus~$W = VD \in R^n$. It follows that~$Z = PU$ also generates a
\sd Picard-Vessiot extension, as claimed in the above.
\end{proof}

\begin{remarks}\label{remark2.13}
If $K$ is a maximal Picard-Vessiot extension of $E$ with the same
constants (Zorn's lemma guarantees that such fields exist), then the
hypothesis on the existence of $\overline{Z}$ in
Proposition~\ref{PVseq} is always satisfied. Therefore for such a
field~$K$, $\calS_K$ contains a \sd Picard-Vessiot ring for any
system~$\{\sigma(Y)=AY, \, \delta(Y)=BY\}$.
\end{remarks}

\subsection{Liouvillian solutions}%
The Galois theory for linear differential equations  is stated in
terms of differential integral domains and fields \cite[Chapter
1]{putsinger2} and  both theory and algorithms for finding
liouvillian solutions are well developed~\cite[Chapters 1.5, 4.1 -
4.4]{putsinger2}.  The main result is that the associated
Picard-Vessiot extension lies in a tower of fields built up by
successively adjoining, exponentials, integrals and algebraics if
and only if the associated Galois group has a solvable identity
component. For linear difference equations, the Galois theory is
stated in terms of reduced rings and total rings of fractions. A
general theory of liouvillian solutions has not been developed in
the difference case. However, a case has been investigated
in~\cite{hendrikssinger} where the coefficient field is of the
form~$C(x)$ with a shift operator~$\sigma:x\mapsto x+1$ and~$\sigma
|_C = \id$. In this situation, solutions of linear difference
equations are identified with sequences whose entries are in $C$.
One says that a linear difference equation is solvable in terms of
liouvillian sequences if it has a full set of solutions in a ring of
sequences built up by successively adjoining to $C$ sequences
representing indefinite sums, indefinite products and interlacings
of previously defined sequences.  The main result is that a linear
difference equation is solvable in terms of liouvillian sequences if
and only if its Galois group has a solvable identity
component.\\[0.1in]
{In this paper we will combine the approaches for differential and
difference cases to investigate the solvability of} systems of mixed
linear difference-differential equations over $E(x)$  where~$ E$
will always be an algebraically closed differential field unless
specified otherwise, $ \sigma(x) = x+1$
and~$\sigma|_E =\id$. \\[0.1in]
In this section, we will give a characterization of Galois groups
for mixed difference-differential systems
to have solvable identity component in terms of liouvillian towers
over an arbitrary \sd-field~$k$ with algebraically closed constants.
Then  we will define a notion of liouvillian sequences and show that
having a full set of solutions of this type implies that the Galois
group has solvable identity component (Proposition~\ref{solvseq}).
In a later result (Proposition~\ref{solnliouvseq}),
we will  show  the converse is true as well. \\[0.2in]
 Liouvillian extensions for
{\it \sd fields}  are defined in the usual way.
\begin{define}
Let $k$ be a \sd field. A \sd field extension $K$ of $k$ is said to
be {\rm liouvillian} if there is a chain of \sd field extensions
$$
k=K_0\subset K_1\subset \cdots \subset K_m{= K}
$$
such that~$k^{\sigma\delta}=K^{\sigma\delta}$, {\ie, $K$
shares the same set of constants with $k$,} and~$K_{i+1}=K_i(t_i)$
for~$i=0,\dots,m-1$ where
\begin{enumerate}
   \item [$(1)$]
        $t_i$ is algebraic over $K_i$, or
    \item [$(2)$]
         $\sigma(t_i)=r_1t_i$ and~$\delta(t_i)=r_2t_i$ with~$r_1,r_2 \in K_i$ \ie, $t_i$ is
         hyperexponential over $K_i$,   or
    \item [$(3)$]
         $\sigma(t_i)-t_i \in K_i$ and~$\delta(t_i) \in
        K_i$.
\end{enumerate}
\end{define}
We now define liouvillian solutions of mixed difference-differential
systems.  In the sequel, let $k$ be a \sd-field with algebraically
closed constants, $R$ be a~\sd~PV extension of $k$ for the system  $
\{\sigma(Y)=AY, \,\delta(Y)=BY\} $ with~$A \in \mat_n(k)$ and $B \in
\mathfrak{g}l_n(k)$. Suppose that~$R$ has a decomposition
$$
R=e_0R\oplus e_1R \oplus \cdots\oplus e_{s-1} R
$$
and $\mathcal{F}_0$ is the quotient field of $e_0R$. Then
$\mathcal{F}_0$ is a~$\sigma^s\delta$-field.

\begin{define}\label{DEF:liouv}
Let~$v=\sum_{i=0}^{s-1}v_i$ be a solution  of~$ \{\sigma(Y){=}AY,
\,\delta(Y){=}BY\} $ in $R^n$ where $v_i:=e_iv \in {e_i R^n}$. We
say that $v$ is {\rm liouvillian} if {the entries of $v_0$ lie in
a~$\sigma^s\delta$-liouvillian extension of~$k$ containing
$\mathcal{F}_0$}. We say that the original system is {\rm solvable
in liouvillian terms} if each
 solution is liouvillian.
\end{define}
{Suppose that $v$ is a solution of~$ \{\sigma(Y)=AY,
\,\delta(Y)=BY\} $ in $R^n$.}
From
$$
\sigma (v) =Av\quad\mbox{and}\quad \delta (v) =B v,
$$
it follows that
\begin{eqnarray*}
\sigma(v_{s-1})\oplus\sigma(v_0)\oplus \cdots \oplus
\sigma(v_{s-2})&=&Av_0\oplus Av_1\oplus \cdots \oplus Av_{s-1}\\
\delta(v_0)\oplus \delta(v_1)\oplus \cdots \oplus \delta(v_{s-1})&=&
Bv_0\oplus Bv_1\oplus\cdots \oplus Bv_{s-1},
\end{eqnarray*}
which implies that~$ \sigma(v_i)=Av_{i+1 \,\mbox{mod} \, s}$
and~$\delta (v_i)=B v_i $ for~$i=0,\dots, s{-}1$. Hence
$\sigma^{s}(v_0)=A_sv_0$ and $\delta(v_0)=Bv_0$, \ie, $v_0$ is a
solution of the system~$\{\sigma^s(Y)=A_s Y,\, \delta(Y)=BY\}.$
Conversely, assume that $v_0$ is a solution of $\{\sigma^s(Y)=A_s
Y,\, \delta(Y)=BY\}$ in $e_0R^n$. Let $v_i=A^{-1}\sigma(v_{i-1})$
  for $i=1,\cdots,s-1$ and $v=v_0\oplus \cdots \oplus v_{s-1}$. We then have
  $\sigma(v)=Av$. From the fact that $\sigma(B)-\delta(A)A^{-1}=ABA^{-1}$, one can easily check that
   $\delta(v)=Bv$. Hence $v$ is a solution of $ \{\sigma(Y)=AY,
\,\delta(Y)=BY\} $ in $R^n$.
 Moreover we have the following
\begin{prop}
\label{prop1} The system~$ \{\sigma(Y)=AY, \,\delta(Y)=BY\} $ is
solvable in liouvillian terms if and only if the
system~$\{\sigma^s(Y)=A_s Y,\, \delta(Y)=BY\}$ is solvable in
liouvillian terms.
\end{prop}
\begin{proof}
Let $V_1,\cdots, V_n$ be a basis of the solution space for the
system $ \{\sigma(Y)=AY, \,\delta(Y)=BY\}$. Since $V_i$ is
liouvillian, $V_{i0}=e_0V_i$ is liouvillian for each~$i$. It then
suffices to show that $V_{10},\cdots, V_{n0}$ are linearly
independent over $k^{\sigma\delta}$. Assume that there exist
$c_1,\cdots,c_n \in k^{\sigma\delta}$, not all  zero, such that
$c_1V_{10}+\cdots+c_nV_{n0}=0$. Letting $V_{i1}=e_1V_i$ we
have~$V_{i1}=A^{-1}\sigma(V_{i0})$ for~$i=1,\dots, n$. Remark that
$k^{\sigma\delta}=k^{\sigma^{s}\delta}$ since $k^{\sigma\delta}$ is
algebraically closed.~Therefore
$$c_1V_{11}+\cdots+c_nV_{n1}=A^{-1}\sigma(c_1V_{10}+\cdots+c_nV_{n0})=0.$$
Similarly,  $c_1e_iV_1+\cdots+c_ne_iV_n=0$ for each~$i$. Hence
$c_1V_1+\cdots+c_nV_n=0$, a contradiction.\\[0.1in]
Conversely, suppose that $V_{10},\cdots,V_{n0}$ is a basis of the
solution space for the system~$\{\sigma^s(Y)=A_s Y,\,
\delta(Y)=BY\}$ and that all the $V_{i0}$'s are liouvillian.
For~$i=1,\cdots,n$ and~$k=1,\cdots,s-1$, let
$$
V_{ik}=A^{-1}\sigma(V_{ik-1})\quad\mbox{and}\quad V_i=V_{i0}\oplus
V_{i1}\oplus \cdots \oplus V_{i,s-1}.
$$
Then each $V_i$ is a solution of~$ \{\sigma(Y)=AY, \,\delta(Y)=BY\}
$. Clearly, $V_1,\cdots,V_n$ are linearly independent over
$k^{\sigma\delta}$. This concludes the proposition.
\end{proof}

\begin{theorem}
\label{thm1}The system~$ \{\sigma(Y)=AY, \,\delta(Y)=BY\}$ is
solvable in liouvillian terms if and only if $\gal(R/k)^0$ is
solvable.
\end{theorem}
\begin{proof}
By Proposition \ref{prop1},~$ \{\sigma(Y)=AY, \,\delta(Y)=BY\}$  is
solvable in liouvillian terms if and only if the associated
system~$\{\sigma^s(Y)=A_s Y,\, \delta(Y){=}BY\}$  is solvable in
liouvillian terms. By Lemma \ref{lem4}, $\gal(e_0R/k)$ is a subgroup
of~$\gal(R/k)$ of finite index. Then $\gal(R/k)^0=\gal(e_0R/k)^0$.
Hence it suffices to show that the associated system  is solvable in
liouvillian terms if and only if $\gal(e_0R/k)^0$ is solvable. Note
that the PV extension of $k$ for the associated system  is a domain.
Thus the proof is  similar to that in differential  case.
\end{proof}
 Let~$E$ be a differential field with algebraically closed constants and $L$
be a field containing~$E$ which satisfies the following conditions:
\begin{itemize}
\item[(i)] $L$ is a differential field extension of $E$ having the same field of
constants as  $E$;
\item[(ii)] every element in $L$ is liouvillian over $E$;
\item[(iii)] $L$ is maximal with respect to (i) and (ii).
\end{itemize}
Zorn's Lemma guarantees that such a field exists. We refer to~$L$ as
a {\it maximal liouvillian extension of $E$}.  One can show that any
differential liouvillian extension of $E$ having the same field of
constants as $E$ can be embedded into a maximal liouvillian
extension and that any two maximal liouvillian extensions of $E$ are
isomorphic over $E$ as differential fields. \\[0.1in]
Recall  that {for two sequences~$\bfa$ and~$\bfb$, and
for a nonnegative integer~$m$,~$\bfb$ is called the {\it i}th
$m$-interlacing (\cite[Definition 3.2]{hendrikssinger}) of $\bfa$
with zeroes if
$$
\bfb(mn+i) = \bfa(n) \quad\mbox{and}\quad \bfb(r) = 0\quad\mbox{for
any~$r\not \equiv i \mod m$}.
$$
A sequence~$\bfa$ is called the $i$th $m$-section of $\bfb$
if~$\bfa(mn+i)=\bfb(mn+i)$ and $\bfa(r)=0$ for~$r \not\equiv i \mod
m$.}\\[0.1in]
We now turn to a definition of liouvillian sequences.
\begin{define}\label{liouvseqdef}
Let $E$ be a differential field with  algebraically closed
constants, {
$\sigma$ be an automorphism on~$E(x)$ satisfying~$\sigma(x)=x+1$
and~$\sigma|_E =\id$.} Let  $L$ be a maximal liouvillian extension
of $E$. The {\rm ring of liouvillian sequences} over $E(x)$ is the
smallest subring $\calL$ of $\calS_L$ such that
\begin{enumerate}
\item[(1)]  ${L(x)} \subset \calL$;
\item[(2)]  For~$\bfa\in {\cal S}_L$, $\bfa \in \calL$ if and only if $\sigma(\bfa) \in
\calL$;
\item[(3)] Supposing~$\sigma(\bfb) =
\bfa\bfb$ with~$\bfa, \bfb\in \calS_L$, then $\bfa \in E(x)$ implies
$\bfb \in \calL$. {$\bfb$ is called a hypergeometric sequence over
$E(x)$;}
\item[(4)] {Supposing~$\sigma(\bfb) = \bfa
+\bfb$ with~$\bfa, \bfb\in \calS_L$, then~$\bfa \in \calL$ implies
$\bfb \in \calL$};
\item[(5)] For~$\bfa\in {\cal S}_L$, $\bfa \in \calL$ implies that $\bfb \in \calL$, where {$\bfb$ is the {\it i}th $m$-interlacing
of $\bfa$ with zeroes  for some  $m\in \bZp$ and $0\leq i \leq
m-1$.}
\end{enumerate}
\end{define}
Set~$\tilde{k}=E(x)$ and let~$\calL$ be the ring of liouvillian
sequences over $\tilde{k}$. From the remarks in~\cite[p.
243]{hendrikssinger}, if $\bfb\in \calS_L$ belongs to~$\calL$ then
the~$i$th~$m$-section of $\bfb$ also belongs to~$\calL$ for any $i$
and $m$. We claim that $\calL$ is a~\sd ring.
 Since~$\calL$ can be constructed inductively using (1) - (5) above, it is enough to show the following statements:
\begin{enumerate}
\item[(1')]  If $ f\in {L(x)}$ then $\delta(f) \in {L(x)}$;
\item[(2')]  If~$\bfa\in {\cal S}_L$,  {$\bfa, \delta(\bfa) \in \calL$ then $\delta(\sigma(\bfa)) \in
\calL$};
\item[(3')]Supposing~$\sigma(\bfb) =
\bfa\bfb$ with~$\bfa, \bfb\in \calS_L$, then $\bfa, \delta(\bfa)\in
E(x)$ implies $\delta(\bfb) \in \calL$;
\item[(4')] {Supposing~$\sigma(\bfb) = \bfa
+\bfb$ with~$\bfa, \bfb\in \calS_L$, then~$\bfa, \delta(\bfa) \in \calL$ implies
$\delta(\bfb) \in \calL$};
\item[(5')] For~$\bfa\in \calS_L$, $\bfa, \delta(\bfa) \in \calL$ implies that $\delta(\bfb) \in \calL$, where $\bfb$ is the {\it i}th $m$-interlacing
of $\bfa$ with zeroes  for some  $m\in \bZp$ and $0\leq i \leq m-1$.
\end{enumerate}
Verifying $(1'), (2')$ an $(5')$ is straighfoward. To verify $(3')$, suppose
that~$\bfb\in\calL$ with~$\sigma(\bfb)=\bfa\bfb$. Set~$y =
\delta(\bfb)$. Then~$\sigma(y) - \bfa y = \delta(\bfa) \bfb$.
Since~$\bfa \in \tilde{k}$, $\bfb$ is invertible. Then~$\bfu :=
\delta(\bfb)/\bfb$ satisfies~$\sigma(\bfu) - \bfu =
\delta(\bfa)/\bfa$.  Therefore~$\bfu \in \calL$ and so $\delta(\bfb)
= \bfu \bfb \in \calL$.  To verify $(4')$, suppose that~$\bfb\in
\calL$ with~$\sigma(\bfb) = \bfa +\bfb$. Then~$\sigma(\delta(\bfb))
=
\delta(\bfb) + \delta(\bfa)$ which implies that $\delta(\bfb)\in \calL$.\\[0.1in]
A vector is said to be {\it hypergeometric} over $\tilde{k}$ if it
can be written as~$Wh$ where~$W \in \tilde{k}^n$ and $h$ is a
hypergeometric sequence over $\tilde{k}$.}
 For any positive integer $d$, we can construct a solution
of~$\{\sigma^d(Y)=A_dY,\, \delta(Y)=BY\}$  by interlacing as indicated below. Consider a set of  new
systems
\begin{equation*}
\label{interlacing}
  \mathbb{S}_j: \,\,\, \, \{\sigma(Z)=A_d(dx+j)Z,\quad
  \delta(Z)=B(dx+j)Z\}, \quad j=0,\cdots,d-1,
\end{equation*}
where 
$A_d(dx+j)$ and $B(dx+j)$ mean replacing $x$ by $dx+j$ in each entry
of~$A_d$ and $B$, respectively. Clearly, these new systems are all
integrable. Moreover, we have the following
\begin{prop}
\label{prop-interlacing} Let $V_j$ be a solution of $\mathbb{S}_j$
in $\calS_K^n$ {where $\calS_K$ is as in Remark
\ref{remark2.13}} and~$W_j$ be the $j$th $d$-interlacing with zeros
of $V_j$ for $j=0,\cdots,d-1$. Then
$$
W=W_0+\cdots+W_{d-1}
$$ is a solution of
$\{\sigma^d(Y)=A_dY,\, \delta(Y)=BY\}.$
\end{prop}
\begin{proof}
It suffices to show that~$W_j$ is a solution
of~$\{\sigma^{d}(Y){=}A_dY,\,\delta(Y){=}BY\}$ for each~$j$. Let
$V_j=(V_j(0),V_j(1),\cdots)$ and~$W_j=(W_j(0),W_j(1),\cdots)$ for
$j=0, \cdots, d-1$. Then
$$
V_j(i+1)=A_d(di+j)V_j(i), \quad
   \delta(V_j(i))=B(di+j)V_j(i)
$$
and therefore by definition of interlacing, for $i\geq 0$ and
$\ell=1,\cdots,d-1$,
\begin{align*}
  &W_j(di+j+d)=V_j(i+1)=A_d(di+j)V_j(i)=A_d(di+j)W_j(di+j),\\
  &W_j(di+j+d+\ell)=0=A_d(di+j)W_j(di+j+\ell).
\end{align*}
So $\sigma^d(W_j)=A_dW_j$. It is clear that $\delta(W_j)=BW_j$. So
the proposition holds.
\end{proof}
\begin{define}\label{DEF:solvable_sys}
Let~$E$ be an algebraically closed differential field with a
derivation~$\delta$. Suppose that~$\delta$ is extended to be a
derivation~$\delta$ on~$E(x)$ such that~$\delta(x) = 0$ and~$\sigma$
is an automorphism on~$E(x)$ such that~$\sigma|_E = \id$
and~$\sigma(x) = x+1$. A system
 \begin{equation*} 
\sigma(Y) = AY, \quad\delta(Y) = BY
\end{equation*}
over~$E(x)$ is said to be {\rm solvable in  terms of liouvillian
sequences} if the \sd PV extension of this system embeds, over
$E(x)$, into~$\calL$,  the ring of liouvillian sequences
over~$E(x)$.
\end{define}
Clearly, Definition~\ref{DEF:solvable_sys} generalizes the notion of
{solvability} in liouvillian terms for linear differential equations
and that for linear difference equations. {In addition, we shall
show in Proposition~\ref{solvseq} that this property is also
equivalent to the identity component of the Galois group being
solvable. However this property is not equivalent to having a
solution in the ring $\calS_L$, as shown by the following
\begin{example} The Hermite polynomials \[H_n(t) = n!\sum_{m=0}^{[n/2]}
\frac{(-1)^m(2t)^{n-2m}}{m! (n-2m)!}\]  satisfy a linear differential
equation with respect to $t$ and a difference equation with respect to $n$
(cf., \cite{bateman}, Ch.~10.13).  In matrix terms, the vector $Y(n,t) = (H(n,t), H(n+1,t))^T$ satisfies
\begin{eqnarray*}\label{hermitesys}
Y(n+1,t)  &= &\left(\begin{array}{cc}0&1\\-2n& 2t\end{array}\right) Y(n,t), \\
\frac{\partial Y(n,t)}{\partial t} & = & \left(\begin{array}{cc}2t&-1\\2n& 0\end{array}\right) Y(n,t)
\end{eqnarray*}
so this system has a solution in $\calS_L$ and one can use the
techniques of Section~\ref{sec6} to show that this system have no
solutions in $\calL$.
\end{example} }

\noindent Concerning total rings of fractions, we have the following

\begin{lemma}
Let $k$ be a \sd-field  and $S$ a \sd-PV extension of $k$.  A
nonzero element $r \in S$ is a zero divisor in $S$ if and only if
there exists  $j\in \bZp$ such that $\prod_{i=0}^j \sigma^i(r) = 0$.
\end{lemma}
\begin{proof} Suppose that $\prod_{i=0}^j \sigma^i(r) = 0$ for
some $j\in\bZp$ and let $j$ be minimal with respect to this
assumption. If $0 = \prod_{i=1}^j \sigma^i(r) =
\sigma(\prod_{i=0}^{j-1} \sigma^i(r)) $, then we would also
have~$\prod_{i=0}^{j-1} \sigma^i(r) = 0$ contradicting the
minimality of $j$.  Therefore~$r$ is a zero divisor. Now suppose
that $r$ is a zero divisor. We  write~$S = \oplus_{i=0}^{j-1} S_i$
where the $S_i$ are domains and $\sigma(S_i )= S_{i+1\!\!\mod j}$.
Then~$r= \sum_{i=0}^{j-1}r_i$ with~$r_i\in S_i$. Writing~$r = (r_0,
\ldots , r_{j-1})$, one sees that $r$ is a zero divisor in $S$ if
and only if some $r_i$ is zero. Assume~$r_0 = 0$. Then
$$
\sigma(r) = (\sigma (r_{j-1}), 0, \ldots  ), \quad \sigma^2(r) =
(\sigma^2(r_{j-2}), \sigma^2(r_{j-1}), 0 , \ldots ), \quad  \ldots
$$
so $0 = \prod_{i=0}^{j-1}\sigma^i(r)$.
\end{proof}
Consequently, we have
\begin{cor} \label{embed} Let $k$ and $S$ be as above and let $\bar{S}$ be a $\sigma$-subring of $S$.
An element $r \in \bar{S}$ is a zero divisor in $\bar{S}$ if and
only if it is a zero divisor in $S$. Therefore the  total ring of
fractions of $\bar{S}$ embeds in the total ring of fractions of $S$.
\end{cor}

\begin{prop}\label{solvseq}
Let $E(x)$ be as in Definition~\ref{DEF:solvable_sys}. If a system
\begin{eqnarray*}
\sigma(Y) = AY, \quad  \delta(Y) = BY
\end{eqnarray*}
over~$E(x)$ is solvable in terms of liouvillian sequences, then the
identity component of its Galois group is solvable.
\end{prop}
\begin{proof}
Let $R$,~$L$ and~$\calL$ be the \sd PV ring, a maximally liouvillian
extension of~$E$ and the ring of liouvillian sequences for the given
system, respectively. Then~$R\subset \calL$. Consider the following
diagram
 \begin{center}
$RL(x)$\\
$/ \hspace{.2in} \backslash$\\
$R \hspace{.3in}L(x)$\\
$\backslash \hspace{.2in} /$\\
$R\cap L(x)$\\
$\mid$\\
$E(x)$
\end{center}
where $RL(x)$ is the field generated by $R$ and $L(x)$. We will show
that
\begin{enumerate}
\item[(i)] $RL(x) \subset \calL$ is a \sd PV extension of $L(x)$
and its Galois group has a solvable identity component.
\item[(ii)] The \sd Galois group of $RL(x)$ over $L(x)$ is isomorphic to the
subgroup~$H$ of the Galois group $G$ of $R$ over {$E(x)$}
that leaves the quotient field of $R\cap L(x)$ fixed. Moreover,~$H$
is a closed normal subgroup of $G$.
\item[(iii)] $G/H$ is the Galois group of the quotient field  of $R\cap L(x)$ over $E(x)$ and
has solvable identity component.
\end{enumerate}
Once the above claims are proven, the group~$G$ has a solvable
identity component since both~$H$ and~$G/H$ have solvable
identity component. \\[0.1in]
To prove (i), we consider~$\{\sigma(Y) = AY, \, \delta(Y) = BY\}$ as
a system over~$L(x)$. Since~$R \subset \calL$, then for a
sufficiently large $N$ there is a fundamental matrix of $\delta(Y) =
B(N)Y$ with entries in $L$. Applying Proposition~\ref{PVseq}, we
conclude that $\calS_L$ contains a~\sd PV extension $T$ of~$L(x)$
for the system~$\{\sigma(Y) = AY, \, \delta(Y) = BY\}$ and that
$R\subset T$. This implies that~$RL(x)$ is a \sd PV extension of
$L(x)$. Proposition~4.1 of~\cite{putsinger1} implies that $RL(x)$ is
also a difference Picard-Vessiot extension for~$\sigma(Y) = AY$ and
the results of \cite{hendrikssinger} then imply that the {\it
difference} Galois group has solvable identity component. The \sd
Galois group is a subgroup of this latter group and its identity
component
is a subgroup of the identity component of the larger group.  Therefore it is also solvable.\\[0.1in]
To prove (ii), let $F_L$ be the total ring of fractions of $RL(x)$
and $F_E$  the total ring of fractions of $R$. Corollary~\ref{embed}
 implies that we can regard~$F_E$ as a subset of~$F_L$.  The
elements of the \sd Galois group $\bar{G}$ of $F_L$ over $L(x)$
restrict to automorphisms of $F_E$ over $E(x)$, and this gives a
homomorphism of this group into $H$.  Clearly the image is closed
and the set of elements  left fixed by this group is $L(x) \cap
F_E$, the quotient field of $L(x) \cap R$.
Therefore this image is~$H$.  \\[0.1in]
Before  proving that $H$ is normal in $G$, we first show that
$$
F_E\cap L(x) = (F_E\cap L)(x).
$$
Since $x \in F_E$ we have~$F_E\cap L(x) \supset (F_E\cap L)(x)$. To
get the reverse inclusion, let $f \in F_E \cap L(x)$ and write
\[f = \frac{a_rx^r + \ldots +a_0}{b_sx^s+ \ldots +b_0},\,\, a_i,b_i \in L\]
where the numerator and denominator are relatively prime.  We then
have that $\{x^sf, x^{s-1}f, \ldots , f, x^r, \ldots , 1\}$ are
linearly dependent over $L$.  Since $L$ is the set of
$\sigma$-invariant elements of $L(x)$,  the cassoratian of these
elements must vanish (\cite{cohn}, p.271). This  implies further
that these elements are linearly dependent over the
$\sigma$-invariant elements of the field~$F_E \cap L(x)$. Therefore
there exist $\sigma$-invariant elements $\tilde{a}_i$,~$\tilde{b}_j
\in F_E \cap L$ for~$i=0,1,\dots, r$ and~$j=0, 1,\dots, s$, not all
zero, such that
$$
\tilde{a}_rx^r + \cdots
+\tilde{a}_0-(\tilde{b}_sx^sf+ \cdots +\tilde{b}_0f)=0.
$$  Since $x$
is transcendental over $\sigma$-invariant elements, there exists at
least one $\tilde{b}_i$ which is not zero. Hence
\[f = \frac{\tilde{a}_rx^r + \ldots \tilde{a}_0}{\tilde{b}_sx^s+ \ldots +\tilde{b}_0} \in (F_E \cap L)(x) .\] \\[0.1in]
To show that $H$ is normal, it now suffices to prove that any
$\sigma\delta$-automorphism of $F_E$ over $k$ leaves the field $F_E
\cap L(x) = (F_E \cap L)(x) $ invariant. Note that $L$ is the set of
$\sigma$-invariant elements of $F_L$ and so $F_E \cap L $ is the set
of $\sigma$-invariant elements of $F_E$.
This set is clearly preserved by any \sd-automorphism.\\[0.1in]
To prove (iii), note that since $H$ is normal, the {\it field}
$(F_E\cap L)(x)$ is a \sd-PV extension of $E(x)$. Furthermore,
$(F_E\cap L)(x)$ lies in the liouvillian extension~$L(x)$ of $E(x)$.
Theorem~\ref{thm1} implies that its Galois group is
solvable.~\end{proof}
{Later in Proposition~\ref{solnliouvseq} we will show
that the converse is true as well.}

\subsection{From $e_0R$ to $R$}\label{sec2}
Throughout this section, {let~$E$ be  an algebraically closed
differential field with a derivation~$\delta$ whose extension
on~$E(x)$ satisfies~$\delta(x) = 0$, and let~$\sigma$ be an
automorphism on~$E(x)$ such that~$\sigma|_E =\id$ and~$\sigma(x) =
x+1$. Set~$\tilde{k} =E(x)$ and let~$\calL$ be the ring of
liouvillian sequences over~$\tilde{k}$.} In this section, we shall
show how to construct a solution in $\calL^n$ of $\{\sigma(Y)=AY,
\delta(Y)=BY\}$ from a solution in $\calL^n$ of its associated
system~$\{\sigma^d(Y)=A_dY, \delta(Y)=BY\}$. Moreover, we shall
prove that {if the associated system
is equivalent over $\tilde{k}$ to a diagonal form,  so is the
original
system 
under the assumption that $A$ is of particular form.}\\

\noindent Let $V \in \calL^n$ be a nonzero solution of
$\{\sigma^d(Y)=A_dY, \delta(Y)=BY\}$ and~$N\in \bZp$ be such that
$V(N)\neq 0$, $A(j)$ and $B(j)$ are well defined and~$\det(A(j))\neq
0$ for $j \geq N$. We define a vector $W$ in the following way:
$$
W(N)=V(N)\quad\mbox{and}\quad W(j+1)=A(j)W(j)\quad\mbox{ for $j\geq
N$}.
$$
Since $\delta(V(N))=B(N)V(N)$, the integrability condition on
$\sigma$ and $\delta$ implies that $W$ is a nonzero solution of
$\{\sigma(Y)=AY, \delta(Y)=BY\}$. The proposition below says that
$W$ is also in $\calL^n$.
\begin{prop}
\label{prop11}Let $W$ be as above. Then
  $W \in \calL^n$.
\end{prop}
\begin{proof}
Let $N=d\ell+m$ where $\ell, m\in\bZp, 0<m<\ell$ and  $V_0$ be
the~$m$th~$d$-section of~$V$. Then $V_0$ is a solution of
$\{\sigma^d(Y)=A_dY, \delta(Y)=BY\}$ and $V_0 \in \calL^n$. Let
$$
V_i(j)=A(j)^{-1}V_{i-1}(j+1)\quad\mbox{for $j \geq N$ and
$i=1,\cdots,d-1$},
$$
and $U=V_0+V_1+\cdots+V_{d-1}.$ Then $U\in \calL^n$. We shall prove
that $ W=U. $ Note that for $j>N$,
$$
V_i(j)=A(j)^{-1}A(j+1)^{-1}\cdots A(j+i-1)^{-1}V_0(j+i).
$$
In particular, $V_{d-1}(j+1)=A(j)V_0(j)$. Then~$V_i(N)=0$
for~$i=1,\cdots,d-1$. Therefore~$
W(N)=V_0(N)+V_1(N)+\cdots+V_{d-1}(N)=U(N)$ and
\begin{align*}
    U(j+1)&=V_0(j+1)+V_1(j+1)+\cdots+V_{d-1}(j+1)\\
          &=A(j)V_1(j)+A(j)V_2(j)+\cdots+A(j)V_0(j)=A(j)U(j)
\end{align*}
for $j\geq N$. Hence $W=U\in \calL^n $.
\end{proof}
{\begin{lemma} \label{prop5} Assume that $w\in \tilde{k}$ satisfies
\begin{equation*}
\label{eqn3} \sigma^s(w)=\sigma^{s-1}(b)\cdots\sigma(b)b w
\end{equation*}
where~$s\in \bZp$, $b\in \tilde{k} \setminus \{0\}$ and~$
b=x^\nu+b_1x^{\nu-1}+\cdots$ with~$\nu\in \mathbb{Z}$ and~$b_i \in
E$. Then $\sigma(w)=bw$.
\end{lemma}
\begin{proof}
We have
\begin{align*}
  \sigma^s\left(\frac{\sigma(w)}{b}\right)
&=\frac{1}{\sigma^{s}(b)}\sigma\left(\sigma^{s-1}(b)\cdots\sigma(b)bw\right)
=\sigma^{s-1}(b)\cdots\sigma(b)b\frac{\sigma(w)}{b}.
\end{align*}
Since $\frac{\sigma(w)}{b}\in \tilde{k}$, $\frac{\sigma(w)}{b}=cw$
for some $c \in \tilde{k}^{\sigma^s}=\tilde{k}^{\sigma}$. Note that
$w$ and~$b$ are rational functions in $x$. Expanding $w$ and $b$ as
Laurent series at~$x=\infty$. By comparing the coefficients, we get
$c=1$, so~$\sigma(w)=bw.$
\end{proof}

\begin{theorem}\label{nprop1}
Let~$A=\diag(a_1,\cdots,a_n)$ where $a_i \in \tilde{k}\setminus
\{0\}$ and~$a_i=cx^{\nu_i}+a_{i1}x^{\nu_i-1}+\cdots$ with~$\nu_i \in
\mathbb{Z}$ and~$c,a_{ij}\in E$. Assume that the
system~$\{\sigma^d(Y)=A_dY,\delta(Y)=BY\}$ is equivalent over
$\tilde{k}$ to
$$
\sigma^d(Y)=A_dY, \quad\delta(Y)=\diag(b_1,\cdots,b_n)Y
$$
where $b_i \in \tilde{k}$ for~$i=1,\dots,n$. Then $\{\sigma(Y)=AY,
\delta(Y)=BY\}$ is equivalent over $\tilde{k}$ to
$$
\sigma(Y)=AY, \quad \delta(Y)=\diag(b_1,\cdots,b_n)Y.
$$
\end{theorem}
\begin{proof}
From the assumption, there exists $G \in \mat_n(\tilde{k})$ such
that
$$
 \sigma^d(G)A_d=A_dG,\quad
 G^{-1}BG-G^{-1}\delta(G)=\diag(b_1,\cdots,b_n).
$$
It then suffices to prove that $\sigma(G)A=AG$. Let
$G=(g_{ij})_{n\times n}$. Then
  \[\begin{cases}
     \sigma^d(g_{ii})-g_{ii}=0,
     & i=1,\cdots,n;\medskip\\
     \sigma^d(g_{ij})=
     \sigma^{d-1}\left(\frac{a_i}{a_j}\right)\cdots
     \sigma\left(\frac{a_i}{a_j}\right)\frac{a_i}{a_j}g_{ij},
     & 1 \leq i \neq j \leq n.
  \end{cases}\]
By Lemma \ref{prop5}, $\sigma(g_{ij})=\frac{a_i}{a_j}g_{ij}$ for all
$i,j=1,\cdots,n$. This implies that $\sigma(G)A=AG$.
\end{proof}
}
\subsection{From $\mathbb{C}(x,t)$ to $\overline{\mathbb{C}(t)}(x)$}\label{sec3}
In the following sections, we always let $k_0$ be the \sd field
$\mathbb{C}(t,x)$ with an automorphism $\sigma:x\mapsto x+1$ and a
derivation $\delta=\frac{d}{dt}$ and let $k$ be its extension field
$\overline{\mathbb{C}(t)}(x)$. Consider a  system of
difference-differential equations
\begin{equation*}
  \sigma(Y)=AY, \quad \delta(Y)=BY
\end{equation*}
over $k_0$ where $A \in \mat_n(k_0)$ and~$ B \in \gl_n(k_0).$ We
shall analyze this system by focusing on its difference part
$\sigma(Y) = AY$ and use techniques from the theory of difference
equations. In this latter theory, one assumes that the fixed field
of $\sigma$, that is, the $\sigma$-constants, are algebraically
closed.  For this reason we will need to consider properties
of~$\sigma(Y)=AY$ over $k$ as well as over $k_0$.
We shall first show 
that the \sd Galois group of this system over $k$ can  be identified
with  a normal subgroup of the Galois group of the same system over
$k_0$, and then conclude some results on orders of the factors.
 For example, if the above system is
irreducible over $k_0$, it is possible that the system is reducible
over $k$. In this case, we will prove that the factors of the above
system over $k$ have the same order. A similar result is well known
for differential equations-if one makes a normal algebraic extension
of the base field then the differential Galois group over this new
field is a normal subgroup of the differential Galois group over the
original field and an irreducible equation factors into factors of
equal order. In the mixed difference-differential case or even the
difference case, the fact that Picard-Vessiot extensions may
contain zero divisors introduces some small complication.\\[0.1in]
We start with some lemmas. Let $R$ and $R_0$ be the \sd PV
extensions of $k$ and~$k_0$ for the system~$ \{  \sigma(Y)=AY,
\delta(Y)=BY\}$ respectively.
\begin{lemma}
\label{lem6}
  \begin{enumerate}
     \item [$(i)$]
         There is a {$k_0$-monomorphism  of~$\sigma\delta$-rings from~$R_0$ to~$R$}.
     \item [$(ii)$]
         Identify $R_0$ with a subring of $R$ as in the first assertion. Suppose
         that~$
         R_0=f_0R_0\oplus 
         \cdots \oplus f_{d-1}R_0
         $
         where~$f_iR_0$
         is a domain, $f_i^2=f_i$ and~$\sigma(f_i)=f_{i+1 \, {\rm mod }\, d}$, and that~$R=e_0R \oplus \cdots \oplus e_{s-1}R$ is
         a similar decomposition of $R$. Then $s=md$ for some~$m\in \bZp$. Moreover,  after a possible renumbering of the $f_i$,
         we have
         $$
           f_i=e_i+e_{i+d}+\cdots+e_{i+(m-1)d}\quad \mbox{for
           $i=0,\cdots,d-1$.}
         $$

  \end{enumerate}
\end{lemma}
\begin{proof}
(i) Clearly, the ring $R_0\otimes_{k_0} k$ becomes
a~$\sigma\delta$-ring endowed with the  actions~$ {\delta}(r\otimes
h)={\delta}(r)\otimes h+r\otimes{\delta}(h)$ and~$\sigma(r\otimes
h)=\sigma(r)\otimes\sigma(h) $ for any~$r \in R_0$ and $h \in k$.
Since $k_0$ is a field, the two canonical embeddings
$$
R_0\rightarrow R_0\otimes_{k_0} 1 \subset R_0\otimes_{k_0}k
\quad\mbox{ and }\quad k \rightarrow 1\otimes_{k_0} k \subset
R_0\otimes_{k_0}k
$$
are both injective, and  clearly are homomorphisms
of~$\sigma\delta$-rings. Let $M$ be a maximal~$\sigma\delta$-ideal
of $R_0\otimes_{k_0}k$ and consider the ring
$(R_0\otimes_{k_0}k)/M$. Because~$R_0$ and~$k$ are both
simple~$\sigma\delta$-rings, the above embeddings factor through
to~$(R_0\otimes_{k_0}k)/M$ and are still injective. Note
that~$(R_0\otimes_{k_0}k)/M$ is a \sd-PV extension of $k$
for~$\{\sigma(Y)=AY,  \delta (Y)=BY\}$. So by uniqueness, we may
write~$R = (R_0\otimes_{k_0}k)/M$. Assume
that~$R_0=k_0[Z,\frac{1}{\det(Z)}]$ where $Z$ is a fundamental
matrix of the system. Let $\bar{Z}=Z \mod M$. One sees
that~$\bar{Z}$ is still a fundamental  matrix of the system and
that~$\det(\bar{Z})\neq 0$. Hence
$$
R=\left(k_0\left[Z,\frac{1}{\det(Z)}\right]\otimes _{k_0}k\right)/M
     =k\left[\bar{Z}, \frac{1}{\det(\bar{Z})}\right]
$$
is {a} $\sigma\delta$-PV ring for the given system over $k$ and $R_0$ can be embedded into~$ R$.\\[0.1in]
(ii) Write $f_0=\sum_{j=0}^{s-1}a_je_j$ with $a_j\in e_jR$. Squaring
both sides yields~$f_0=\sum_{j=0}^{s-1}a_j^2e_j$, thus
$a_j^2e_j=a_je_j$. Since $e_jR$ is a domain,~$a_j$ is either~$e_j$
or $0$ for each~$j$. The same 
holds for other $f_i$'s. Then for any~$i=0,\cdots,d-1$, there is a
subset $T_i \subset \{0,\cdots,s-1\}$ such that $f_i=\sum_{j\in
T_i}e_j$. Assume that $T_{i_0}\cap T_{i_1}$ is not empty for two
different~$i_0$ and~$i_1$. Let~$l \in T_{i_0}\cap T_{i_1}$.
Since~$\sum_{i=0}^{d-1}f_i=\sum_{j=0}^{s-1}e_j=1$,
\begin{equation*} \label{eqn4}
0=\sum_{i=0}^{d-1}f_i-\sum_{j=0}^{s-1}e_j=\sum_{i=0}^{d-1}\sum_{j
\in T_i}e_j-\sum_{j=0}^{s-1}e_j=pe_l+H
\end{equation*}
where $p>0$ and $H$ is the sum of all the $e_q$'s with $q \neq l$.
Multiplying both sides of the above equality by $e_l$, we
get~$pe_l=0$, a contradiction. Hence the $T_i$'s form a partition of
$\{0,\cdots,s-1\}$. Since~$\sigma(f_i)=f_{i+1 \, {\rm mod}\,  d}$
and~$\sigma(e_j)=e_{j+1\,{\rm mod}\, s}$, one sees that the sets
$T_i$ have the same size and that a renumbering yields the
conclusion.
\end{proof}
According to Lemma \ref{lem6}, we will always consider $R_0$ as a
subring of~$R$  and assume~$R=k[\bar{Z},\frac{1}{\det(\bar{Z})}]$ in
the sequel. In particular, we can view
$R_0=k_0[\bar{Z},\frac{1}{\det(\bar{Z})}]$.
\begin{lemma}\label{lem7}
Let $\gamma : \gal(R/k)\rightarrow \gal(R_0/k_0)$ be a map given
by~$ \gamma(\phi){=}\phi |_{R_0}$ for any~$\phi \in \gal(R/k)$. Then
$\gamma$ is {a monomorphism}.
Moreover, we can view the identity component
of $\gal(R/k)$ as a subgroup of that of~$\gal(R_0/k_0)$.
\end{lemma}
\begin{proof}
Assume that $R=k[\bar{Z}, \frac{1}{\det(\bar{Z})}]$. Let $\phi \in
\gal(R/k)$. If $\phi(\bar{Z})=\bar{Z}[\phi]_{\bar{Z}}$ for some
$[\phi]_{\bar{Z}} \in \mat_n(\mathbb{C})$, then
$\gamma(\phi)(\bar{Z})=\bar{Z}[\phi]_{\bar{Z}}$. Hence
$\gamma(\phi)$ is  an automorphism of $R_0$ over $k_0$, that is,
$\gamma(\phi)\in \gal(R_0/k_0)$. Note that~$\det(\bar{Z})\neq 0$. If
$\gamma(\phi)=\id$, then $[\phi]_{\bar{Z}}=I_n$, which implies that
$\phi =\id$. So $\gamma$ is an injective homomorphism.  Therefore,
we can view~$\gal(R/k)$ as a subgroup of~$\gal(R_0/k_0)$. Since
$\gamma$ is continuous in the Zariski topology,
$\gamma(\gal(R/k)^0)$ is in~$\gal(R_0/k_0)^0$. So the lemma holds.
\end{proof}
%
%
\begin{lemma}\label{lem8}
  $\gal(R/k)^0=\gal(R_0/k_0)^0$.
\end{lemma}
\begin{proof}
Let $G=\gal(R/k)$ and $G_0=\gal(R_0/k_0)$. From
\cite{hardouin-singer}, the $\sigma\delta$-PV ring $R$ (resp.\
$R_0$) is the coordinate ring of a  $G$-torsor (resp.\
$G_0$-torsor). From~\cite[p.40, (3)]{kunz}, the Krull dimension of
$R$ (resp.\ $R_0$) equals the Krull dimension of $G$ (resp.\ $G_0$).
Since all the components of a linear algebraic group are isomorphic
as varieties, one sees that the Krull dimension of $G$ (resp.\
$G_0$) equals the Krull dimension of $G^0$ (resp.\ $G_0^0$). Since
$R$ is generated over $R_0$ by the elements of $k$, by Proposition
2.2 and Corollary 2.3 in~\cite[p. 44]{kunz}, $R$ is an integral ring
extension of $R_0$. By \cite[Corollary 2.13]{kunz}, the Krull
dimension of $R$ equals that of $R_0$. Hence    the Krull dimension
of~$G^0$ equals that of $G_0^0$. Since both~$G^0$ and $G_0^0$ are
connected and $G^0\subset G_0^0$, by the proposition in \cite[p.
25]{humphreys} we have $G^0=G_0^0$.
\end{proof}
From Lemma \ref{lem7}, $\gal(R/k)$ can be viewed as a subgroup of
$\gal(R_0/k_0)$. In the following, we prove that $\gal(R/k)$ is a
normal subgroup of $\gal(R_0/k_0)$.
Let $\mathcal{F}_0$ and $\mathcal{F}$ be the total ring of
fractions of $R_0$ and $R$ respectively. Note that
$\gal(R/k)=\gal(\F/k)$ and $\gal(R_0/k_0)=\gal(\F_0/k_0).$
Corollary~\ref{embed} allows us to assume that $\mathcal{F}_0
\subset \mathcal{F}$.
\begin{lemma}\label{lem32}
Let $u \in k$ be of degree $m$ over $k_0$ and let $u = u_1, u_2,
\ldots , u_m$ be its conjugates. Then there exist $M \in
\mat_m(k_0)$ and $N \in \gl_m(k_0)$ such~that
\[Z = \left(\begin{array}{cccc}1& 1& \ldots & 1\\ u_1 & u_2 & \ldots &u_m\\
\vdots & \vdots &\vdots &\vdots \\ u_1^{m-1} & u_2^{m-1}& \ldots
&u_m^{m-1}\end{array}\right)  \]  satisfies
\[\sigma(Z) = MZ \quad\mbox{and}\quad    \delta(Z) = NZ.
\]
\end{lemma}
\begin{proof}
We claim that any field automorphism $\tau$ of $k$ over $k_0$ is a
\sd-field automorphism. Indeed, since $k$ is an algebraic extension
of~$k_0$ so any automorphism~$\tau$ is automatically
a~$\delta$-field automorphism. One sees that~$\tau$ is
a~$\sigma$-field automorphism by noting that for any~$f \in
\overline{\mathbb{C}(t)}(x)$,~$\tau$ acts on the coefficients of
powers of $x$ while $\sigma$ acts only on $x$. For any $g$ in the
automorphism group of $k$ over $k_0$, we have $g(Z) = Z [g]$ where
$[g]$ is a permutation matrix.  Since $g$ is also an automorphism
of~\sd-fields, both~$M= \sigma(Z) Z^{-1}$ and $N=\delta(Z)Z^{-1}$
are left invariant by $g$ and therefore must have entries in $k_0$.
\end{proof}
We now proceed to prove the main result of this section.
\begin{prop}
\label{prop10}
 $\gal(R/k)$ is a normal subgroup of $\gal(R_0/k_0)$.
\end{prop}
\begin{proof}
Consider the following diagram
\begin{center}
$\mathcal{F}$\\
$/ \hspace{.2in} \backslash$\\
$\mathcal{F}_0 \hspace{.3in} k$\\
$\backslash \hspace{.2in} /$\\
$\mathcal{F}_0\cap k$\\
$\mid$\\
$k_0$
\end{center}
First, we claim that the map~$\gal(\F/k)\rightarrow
\gal(\F_0/\F_0\cap k)$ that sends~$h \in \gal(\mathcal{F}/k)$ to its
restriction~$h|_{\F_0}$ on $\F_0$ is an isomorphism. Any
automorphism of $\F$ over $k$ is determined by its action on a
fundamental  matrix of~$\{\sigma(Y)=AY, \delta(Y)=BY\}$ and its
restriction on $\F_0$ is determined in the same way.  This implies
that the {restricted} map is injective.  To see that it is
surjective, note that its image is closed and has~$\F_0\cap k$ as a
fixed field. Therefore, the Galois theory
implies that the restricted map  must be~$\gal(\F_0/\F_0\cap
k)$.\\[0.1in]
We now show that~$gh g^{-1}$ is in~$\gal(\F_0/\F_0\cap k)$ for
any~$h\in \gal(\F_0/\F_0\cap k)$ and~$g \in \gal(\F_0/k_0)$. {It
suffices to show that $g$ leaves $\F_0\cap k$ invariant,  which will
imply that~$gh g^{-1}(u) = u$ for any~$u \in \F_0\cap k$ and
so~$ghg^{-1} \in \gal(\F_0/\F_0\cap k)$.
Now let $u \in \F_0\cap k$ be of degree $m$ over $k_0$. From Lemma
\ref{lem32},~$U=(1, u, u^2 , \ldots , u^{m-1})^T$ satisfies some
difference-differential system over $k_0$. Therefore the
vector~$g(U)$  satisfies the same system and so must be a
$\mathbb{C}$-linear combination of the columns of $Z$. In
particular, we have  $g(u) \in k.$  Therefore $g$ leaves $\F_0\cap
k$ invariant. This completes the proof.}
\end{proof}
\begin{theorem}\label{thm6}
If~$\{\sigma(Y)=AY, \delta(Y)=BY\}$ is irreducible over~$k_0$, then
it is equivalent over $\hat{k}_0:=F_0\cap k$ to the system
\begin{equation*}
\sigma(Y)=\diag(A_1,A_2,\cdots,A_d)Y,\quad
\delta(Y)=\diag(B_1,B_2,\cdots,B_d)Y
\end{equation*}
where  $A_i \in \mat_{\ell}(\hat{k}_0), B_i\in
\gl_{\ell}(\hat{k}_0)$ and $\ell=\frac{n}{d}$. Moreover,  the
system~$\{\sigma(Y)=A_iY,\delta(Y)=B_iY\}$ is irreducible over $k$
for~$i=1,\dots, d$, and there exists~$g_i \in \gal(R_0/k_0)$ such
that $g_i(A_1)=A_i$ and $g_i(B_1)=B_i$.
\end{theorem}
\begin{proof}
By Proposition \ref{prop10}, $\gal(R/k)$ is isomorphic to
$\gal(\F_0/\hat{k}_0)$ and then $\gal(\F_0/\hat{k}_0)$ is a normal
subgroup of $\gal(R_0/k_0)$. Let $V$ be the solution space  in
$R_0^n$ of~$\{\sigma(Y)=AY, \delta(Y)=BY\}$. Then Clifford's Theorem
\cite[p.25, Theorem 2.2]{dixon} tells us that $V$ can be decomposed
into $
   V=V_1\oplus V_2\oplus \cdots {\oplus} V_{d}
$ where the $V_i$ are minimal $\gal(\F_0/\hat{k}_0)$-invariant
subspaces of $V$ and, for each~$i$, there exists $g_i\in
\gal(R_0/k_0)$ such that $V_i=g_i(V_1)$. Furthermore,~$g \in
\gal(R_0/k_0)$ permutes the $V_i$. Let $Z_1$ be an $n \times \ell$
matrix whose columns are the solutions  in $V_1$ of the original
system. Then $Z_1$ has the full rank. Then for each $i$, the columns
of $Z_i=g_i(Z_1)$ are the solutions in $V_i$ of the original system
and $Z_i$ has the full rank too. Let $Z=(Z_1,\cdots,Z_d)$. Then~$Z$
is a fundamental matrix of the original system. By Lemma 1
in~\cite{grigoriev}, there exists~$\ell \times n$ matrix $P_1$ of
the rank $\ell$ with entries from $\hat{k}_0$ such that~$P_1V_i=0$
for $i=2,\cdots,d$. Since $g \in \gal(R_0/k_0)$ permutes the $V_i$,
$$
 \{g_i(V_2),\cdots,g_i(V_d)\}=\{V_1,\cdots,V_{i-1},V_{i+1},\cdots,V_d\}.
$$
Let $P_i=g_i(P_1)$. Then $P_iV_j=0$ for $j \neq i$.
Therefore,~$P:=(P_1^T,\cdots,P_d^T)^T\in  \gl_n(\hat{k}_0)$
satisfies that
\begin{equation}\label{eqn31}
PZ=\diag(U_1,U_2,\cdots,U_d)
\end{equation}
with~$U_i \in \gl_\ell(\F_0)$. Moreover  we have $U_i=g_i(U_1)$ for
each $i$. We now prove that~$\det(P) \neq 0.$ Assume the contrary
that~$\det(P)=0$. Then~$w^TP=0$ for some nonzero~$w \in
\hat{k}_0^n$. Therefore there exists $w_i \in \hat{k}_0^\ell$
for~$1\le i\le d$ such that~$w_1^TP_1+\cdots+w_d^TP_d=0.$ Since the
$P_i$ have full rank, there exists at least one~$i$ such
that~$w_i^TP_i\neq 0$. Without loss of generality, assume
that~$w_1^TP_1 \neq 0.$ From~$w_1^TP_1=-(w_2^TP_2+\cdots+w_d^TP_d)$
and~$P_iZ_1=0$ for $i=2,\cdots,d$, we have~$w_1^TP_1Z=0.$ Since
$\det(Z)\neq 0$, $w_1^TP_1=0$, a contradiction. Therefore~$\det(P)
\neq 0$. Let $\ell=\frac{n}{d}$. From
(\ref{eqn31}),~$\{\sigma(Y)=AY, \delta(Y)=BY\}$ is equivalent over
$\hat{k}_0$ to
$$
   \sigma(Y)=\diag(A_1,\cdots,A_d)Y, \quad
   \delta(Y)=\diag(B_1,\cdots,B_d)Y
$$
where $A_i \in\mat_\ell(\hat{k}_0)$ and $B_i \in
\gl_{\ell}(\hat{k}_0).$ Furthermore, $U_i$ is a fundamental matrix
of the system $\{\sigma(Z)=A_iZ,\delta(Z)=B_iZ\}$. Since
$U_i=g_i(U_1)$, we have that $A_i=g_i(A_1)$ and~$B_i=g_i(B_1)$ for
each~$i$. From the minimality of~$V_i$, the
system~$\{\sigma(Z)=A_iZ,\delta(Z)=B_iZ\}$ is irreducible over $k$.
\end{proof}

\begin{cor}
Let~$A_i$ and~$B_i$ be as in Theorem~\ref{thm6} for  $i=1,\dots, d$.
Then for each~$i$,  the Galois group
of~$\{\sigma(Z)=A_iZ,\delta(Z)=B_iZ\}$ over $k$ is solvable by
finite if and only if the Galois group of~$\{\sigma(Y)=AY,
\delta(Y)=BY\}$ over~$k_0$ is  solvable by finite.
\end{cor}
\section{Systems} \label{sec4}
Throughout this section, let $k_0$ be the  field $\mathbb{C}(t,x)$
with an automorphism~$\sigma:x\mapsto x+1$ and  a
derivation~$\delta=\frac{d}{dt}$ and let $k$  the extension field
$\overline{\mathbb{C}(t)}(x)$. In this section, we will first prove
that if a system~$\{ \sigma(Y)=AY, \delta(Y)=BY\}$ with $A \in
\mat_n(k_0)$ and~$ B \in \gl_n(k_0)$ is irreducible over $k_0$ and
its Galois group over $k_0$ is solvable by finite then there
exists~$\ell\in \bZp$ with~$\ell | n$ such that the solution space
of~$\{\sigma^\ell(Y)=A_\ell Y, \delta(Y)=BY\}$ has a basis each of
whose members is the interlacing of hypergeometric solutions over
$k$. We will then refine this result to show
that~$\{\sigma^\ell(Y)=A_\ell Y, \delta(Y)=BY\}$ is equivalent
over~$k$ to a special form. Based on this special form, we will
describe a decision procedure to find its solutions.

\subsection{Systems 
with Liouvillian Sequences as Solutions} By Theorem \ref{thm6}, if a
system $\{ \sigma(Y)=AY,  \delta(Y)=BY\} $ with $A \in \mat_n(k_0)$
and~$ B \in \gl_n(k_0)$ is irreducible over $k_0$, then it  can be
decomposed into factors that  are irreducible over $k$ and if the
Galois group of the system over~$k_0$ is solvable by finite then the
Galois groups of these factors over $k$ are also solvable by finite.
Hence it is enough to consider factors of the original system over
$k$.
\begin{prop}\label{prop2}
Suppose that~$\{\sigma(Y)=\mathcal{A}Y,\delta(Y)=\mathcal{B}Y\}$
with $\mathcal{A}\in \mat_\ell(k)$ and~$\mathcal{B} \in \gl_\ell(k)$
is an  irreducible system over $k$ and that its Galois group
over~$k$ is solvable by finite. Then the system is equivalent over
$k$ to
\begin{equation*}
\sigma (Y) =\bar{\mathcal{A}}Y,\quad \delta (Y) = \bar{\mathcal{B}}
Y
\end{equation*}
where~$\bar{\mathcal{B}} \in \gl_\ell(k)$ and
\[\begin{array}{cccc}
    \bar{\mathcal{A}}=&\begin{pmatrix}
         0 & 1 & 0 & \cdots & 0 \\
         0 & 0 & 1 & \cdots & 0 \\
         \vdots & \vdots & \vdots & \vdots & \vdots \\
         0 & 0 & 0 & \cdots & 1 \\
         a & 0 & 0 & \cdots & 0
      \end{pmatrix}
      \end{array}\in \mat_\ell(k)
\]
with~$a \in k$.
\end{prop}
\begin{proof}
The proof is similar to those of Lemma 4.1 and Theorem 5.1
in~\cite{hendrikssinger}.~\end{proof}
\begin{remark}\label{rem4} From the proof of Lemma 4.1 in
\cite{hendrikssinger}, we know that~$\ell$
divides~$|\gal(\R/k)/\gal(\R/k)^0|$
because~$\{\sigma(Y)=\mathcal{A}Y,\delta(Y)=\mathcal{B}Y\}$ is
irreducible over $k$. From the proof of Theorem 5.1 in
\cite{hendrikssinger},~$\gal(\R/k)^0$ is diagonalizable.
\end{remark}
As a consequence of Proposition~\ref{prop2},  we have the following
\begin{cor}\label{cor12}
If a system~$\{\sigma(Y)=\mathcal{A}Y,\delta(Y)=\mathcal{B}Y\}$ with
$\mathcal{A}\in \mat_\ell(k)$ and~$\mathcal{B} \in \gl_\ell(k)$ is
irreducible over $k$ and its Galois group over $k$ is solvable by
finite, then $\{\sigma^\ell(Y)=\mathcal{A}_\ell Y,
\delta(Y)=\mathcal{B}Y\}$ is equivalent over $k$ to
\begin{equation*}
   \sigma^\ell(Y)=\mathcal{D}Y,
   \quad\delta(Y)=\bar{\mathcal{B}}Y
\end{equation*}
where~$\bar{\mathcal B}\in \gl_\ell(k)$
and~$\mathcal{D}=\diag(a,\sigma(a),\cdots,\sigma^{\ell-1}(a))$ with
$a$ as indicated  in Proposition~\ref{prop2}.
\end{cor}
Next, we shall prove further that $\{\sigma^\ell(Y)=\mathcal{A}_\ell
Y, \delta(Y)=\mathcal{B}Y\}$ is equivalent over $k$ to a system of
diagonal form. Note that  equivalent systems have the same
Picard-Vessiot extension.

\begin{prop}\label{prop41}
If a system~$\{\sigma(Y)=\mathcal{A}Y,\delta(Y)=\mathcal{B}Y\}$ with
$\mathcal{A}\in \mat_\ell(k)$ and~$\mathcal{B} \in \gl_\ell(k)$ is
irreducible over $k$ and its Galois group over $k$ is solvable by
finite, then~$\{\sigma^\ell(Y)=\mathcal{A}_\ell Y,
\delta(Y)=\mathcal{B}Y\}$ is equivalent over $k$ to
$$
\sigma^\ell(Y)=\D Y, \quad \delta(Y)=\diag(b_1,\cdots,b_\ell)Y
$$
with $\D$  as in Corollary~\ref{cor12} and $b_i \in k$
for~$i=1,\dots, \ell$.
\end{prop}
\begin{proof}
By Corollary \ref{cor12}, $\{\sigma^\ell(Y)=\mathcal{A}_\ell Y,
\delta(Y)=\mathcal{B}Y\}$ is equivalent over~$k$
to~$\{\sigma^\ell(Y)=\mathcal{D}Y,
   \delta(Y)=\bar{\mathcal{B}}Y\}$
where~$\bar{\mathcal B}\in \gl_\ell(k)$
and
$$
\mathcal{D}=\diag(a,\sigma(a),\cdots,\sigma^{\ell-1}(a))
$$
with~$a$ as  in Proposition~\ref{prop2}. Let $ \R=\bar{e}_0\R\oplus
\bar{e}_1\R\oplus \cdots \oplus \bar{e}_{\upsilon}\R $ be the
decomposition of~$\R$.
Then~$\gal(\R/k)^0=\gal(\bar{e}_0\R/k)$ by Lemma \ref{lem2}
and~$\gal(\bar{e}_0\R/k)$ is diagonalizable  by Remark~\ref{rem4}.
From Lemma \ref{lem11}, it follows that~$\bar{e}_0\R$ is
a~$\sigma^{\upsilon}\delta$-PV extension of $k$ for the system
\begin{equation}\label{eqn24}
\sigma^{\upsilon}(Y)=\sigma^{\frac{\upsilon}{\ell}-1}(\D)\cdots\sigma(\D)\D
Y,\quad\delta (Y)=\bar{\mathcal{B}}Y.
\end{equation}
Let~$
\tilde{\D}=\sigma^{\frac{\upsilon}{\ell}-1}(\D)\cdots\sigma(\D)\D=\diag(\bar{d}_1,\cdots,\bar{d}_\ell).
$ By Lemma \ref{lem1}, $\bar{e}_0\R$ is a domain. As in the
differential case, we can show that (\ref{eqn24}) is equivalent
over~$k$ to the system
 \begin{equation}\label{eqn67}
\sigma^\upsilon(Y)=\diag(\bar{a}_1,\cdots,\bar{a}_\ell)Y,\quad\delta(Y)=\diag(\bar{b}_1,\cdots,\bar{b}_\ell)Y
\end{equation}
where $\bar{a}_i,\bar{b}_i \in k$. Then there exists $G=(g_{ij}) \in
\mat_\ell(k)$ such that
$$
   \sigma^\upsilon(G)\diag(\bar{a}_1,\cdots,\bar{a}_\ell)=
   \tilde{\D}G,
$$
which implies  $\sigma^\ell(g_{ij})\bar{a}_j=g_{ij}\bar{d}_j.$ Since
$\det(G)\neq 0$, there is a permutation~$i_1,\cdots, i_\ell$ of
$\{1,2,\cdots,\ell\}$ such that $g_{1i_1}g_{2i_2}\cdots g_{\ell
i_\ell}\neq 0$. Hence
$$
\bar{d}_j=\frac{\sigma^\ell(g_{ji_j})}{g_{ji_j}}\bar{a}_{i_j}\quad\mbox{
for $j=1,\cdots,\ell$}.
$$
Let $P$ be a multiplication of some permutation matrices such that
$$P^{-1}\diag(\bar{a}_1,\cdots,\bar{a}_\ell)P=\diag(\bar{a}_{i_1},\cdots,\bar{a}_{i_\ell})$$
and let $T=P\diag(1/g_{1i_1},\cdots, 1/g_{\ell i_\ell}).$ Under the
transformation $Y \rightarrow TY$, the system (\ref{eqn67}), and
therefore (\ref{eqn24}), is equivalent over $k$ to
\begin{equation*}
  \sigma^{\upsilon}(Y)=\tilde{\D}Y,\quad
   \delta(Y)=\diag(b_1,\cdots,b_\ell)Y
\end{equation*}
where $b_i\in k$ for~$i=1,\dots, \ell$. Proposition \ref{nprop1}
implies that the system~$ \{ \sigma^\ell(Y)=\mathcal{D}Y,
   \delta(Y)=\bar{\mathcal{B}}Y
\}$
is equivalent over $k$ to
$$
\sigma^\ell(Y)=\D Y, \quad \delta(Y)=\diag(b_1,\cdots,b_\ell)Y.
$$
This concludes the proposition.
\end{proof}
Theorem \ref{thm6} together with Proposition \ref{prop41} leads to
the following
\begin{prop}\label{thm7}
If the system~$\{\sigma(Y)=AY, \delta(Y)=BY\}$ with~$A \in
\mat_n(k_0)$ and~$B \in \gl_n(k_0)$ is irreducible over $k_0$ and
its Galois group over~$k_0$ is solvable by finite, then there
exists~$\ell\in \bZp$ with~$\ell | n$ such that the solution space
of $\{\sigma^\ell(Y)=A_\ell Y, \delta(Y)=BY\}$ has a basis
consisting of the interlacing of hypergeometric solutions over $k$.
\end{prop}
\begin{proof}
By Theorem \ref{thm6} and Proposition \ref{prop41}, there
is~$\ell\in \bZp$ with~$\ell | n$ such that $\{\sigma^\ell(Y)=A_\ell
Y, \delta(Y)=BY\}$ is equivalent over $k$ to a system of diagonal
form. Since the solution space of the latter system has a basis
consisting of the interlacing of hypergeometric solutions over $k$,
so does the solution space of $\{\sigma^\ell(Y)=A_\ell Y,
\delta(Y)=BY\}$.
\end{proof}
\begin{cor}\label{cor3.6}
Let $\calL$ be the ring of liouvillian sequences over $k$.
Assume that~$\{\sigma(Y)=AY, \delta(Y)=BY\}$ with $A \in
\mat_n(k_0)$ and~$B \in \gl_n(k_0)$ is irreducible over $k_0$ and
its Galois group over $k_0$ is solvable by finite. Then the solution
space of the system has a basis with entries in~$\calL$.
\end{cor}
\begin{proof} Proposition~\ref{thm7} implies that there is~$\ell\in \bZp$ such that the solution space
of~$\{\sigma^\ell(Y)=A_\ell Y, \delta(Y)=BY\}$ has a basis with
entries in~$\calL$. The corollary then follows from Proposition
\ref{prop11}.
\end{proof}
Let us turn to a general case {where a
difference-differential system
may be reducible over the base field. 
If the Galois group over 
the base field of the given system is solvable by finite, then the
Galois group over 
the base field of each factor is of the same type.} The method in
\cite{hendrikssinger} together with the results in \cite{blw}
implies the following

\begin{prop}\label{solnliouvseq}
If the Galois group  for~$\{\sigma(Y)=AY, \delta(Y)=BY\}$ with~$A
\in \mat_n(k_0)$ and~$B \in \gl_n(k_0)$
is solvable by finite, then the solution space of the system has a
basis with entries in $\calL$.
\end{prop}
\begin{proof}
By induction, we only need to prove the proposition for the case
where the given  system has two irreducible factors over $k_0$. In
this case, the given system is equivalent over $k_0$ to
\[
 \sigma(Y) = \left(\begin{array}{cc} A_1 & 0\\ A_3 & A_2\end{array} \right) Y,\hspace{.4in}
\delta(Y) = \left(\begin{array}{cc} B_1 & 0\\ B_3 & B_2\end{array}
\right)Y
\]
where the systems $\{\sigma(Y)=A_iY, \delta(Y)=B_iY\}$ for $i=1,2$
are both irreducible over $k_0$. Let $d_i$ be the order of $A_i$ for
$i=1,2$. By Corollary~\ref{cor3.6}, each system~$\{\sigma(Y)=A_iY,
\delta(Y)=B_iY\}$ has a fundamental matrix~$U_i \in
\mat_{d_i}(\calL)$. From the proof of Theorem 3 in \cite{blw},
Proposition \ref{PVseq} and Remark \ref{remark2.13}, it follows that
the original system has a \sd PV extension $\mathcal{R}$ of $k_0$
which contains entries of the $U_i$'s and, moreover,  has a
fundamental matrix over $\calS_K$  of the form
\[
\begin{array}{cc}
\begin{pmatrix}
     U_1 & 0 \\
     V   & U_2
\end{pmatrix}.
\end{array}
\]
So $\sigma(V)=A_1U_1+A_2V.$ Let $V=U_2W$. Then
$\sigma(W)=W+\sigma(U_2)^{-1}A_1U_1.$ Since $U_i \in
\mat_{d_i}(\calL)$, the entries of $W$ are in $\calL$ and so are the
entries of~$V$.~
\end{proof}
\subsection{Normal Forms}\label{sec5}
Assume that a system $ \{ \sigma(Y)=AY, \ \delta(Y)=BY\} $ where $A
\in \mat_n(k_0)$ and~$B \in \gl_n(k_0)$  is irreducible over $k_0$
and that its Galois group over $k_0$ is solvable by finite. Theorem
\ref{thm6} and Proposition \ref{prop41} imply that there
exists~$\ell\in \bZp$
with~$\ell |n$ such that~$
\{\sigma^{\ell}(Y)=A_\ell Y,  \delta(Y)=BY\}$ is equivalent over $k$
to a system  of diagonal form. In this section, we will show further
that the above system is equivalent over $k$ to a more special form
{when the order~$n$ of the original system is prime}.

\subsubsection{Normal Forms for General {Systems} 
}
Let us first review some notions and properties concerning rational
solutions of difference equations.
\begin{define}
(cf. \cite[{Definition 6.1}]{hardouin-singer}) Let
{$f=\frac{P}{Q}$ with~$P,Q \in
\overline{\mathbb{C}(t)}[x]$ and~$\gcd(P,Q)=1$}.
\begin{enumerate}
     \item [$(1)$]
      The {\rm dispersion} of $Q$, denoted by $\disp(Q)$ is
        $$
          \max\{j \in \bZp |
           Q(\alpha)=Q(\alpha+j)=0\,\, \mbox{for some $\alpha \in
           \overline{\mathbb{C}(t)}$}\}.
        $$
      \item [$(2)$]
        The {\rm polar dispersion} of $f$ is the dispersion of
        $Q$ and denoted  $\pdisp(f)$.
      \item [$(3)$] $f$ is said to be {\rm
standard} with respect to $\sigma^m$, with~$m\in \bZp$, if $\disp(P
\cdot Q)<m$.
   \end{enumerate}
\end{define}
As  in \cite{hardouin-singer}, we have the following
\begin{lemma}
\label{lem10} Assume that $f \in k\setminus \{0\}, a \in
\overline{\mathbb{C}(t)} \setminus \{0\}$ and $m\in \bZp$.
 \begin{enumerate}
   \item [$(1)$]
   There exist $\tilde{f}, \tilde{g} \in k\setminus \{0\}$ such
    that $f=\frac{\sigma^m(\tilde{g})}{\tilde{g}}\tilde{f}$ where
       $\tilde{f}$ is standard with respect to $\sigma^m$.
   \item [$(2)$]
      If $f$ has a pole, then $\pdisp(\sigma^m(f)-af)\geq m.$
  \end{enumerate}
\end{lemma}
\begin{proof} The proof is similar to that {of Lemma 6.2} in
\cite{hardouin-singer}.
\end{proof}

\begin{prop}
\label{prop6} Let $0\neq a,b \in k$  satisfy
$\sigma^m(b)-b=\frac{\delta(a)}{a}$ where $m\in \bZp$. Then
$$
a=\frac{\sigma^m(f)}{f}\alpha(x)\beta(t)\quad\mbox{and}\quad
b=\frac{\delta(f)}{f}+\frac{\delta(\beta(t))}{m \beta(t)}x+c
$$
where $f \in k, c,\beta(t) \in \overline{\mathbb{C}(t)}$, and
$\alpha(x) \in \mathbb{C}(x)$ is standard with respect to
$\sigma^m$.
\end{prop}
\begin{proof}
Let $a=\frac{\sigma^m(f)}{f}\hat{a}$ with $\hat{a}$ standard with
respect to $\sigma^m$ and $\hat{b}=b-\frac{\delta{(f)}}{f}$. Then
$\sigma^m(\hat{b})-\hat{b}=\frac{\delta(\hat{a})}{\hat{a}}$. View
$\hat{a}$ and~$\hat{b}$ as rational functions in~$x$. Then
$\mbox{pdisp}(\frac{\delta(\hat{a})}{\hat{a}})<m$. If
$\frac{\delta(\hat{a})}{\hat{a}} \notin \overline{\mathbb{C}(t)}$,
then $\frac{\delta(\hat{a})}{\hat{a}}$ has a pole and so
does~$\hat{b}$. By Lemma \ref{lem10},
$\mbox{pdisp}(\sigma^m(\hat{b})-\hat{b})\geq m$, a contradiction.
Hence~$\frac{\delta(\hat{a})}{\hat{a}}=w(t) \in
\overline{\mathbb{C}(t)}$, which means
that~$\hat{a}=\alpha(x)e^{\int w(t)dt}$. Since $\hat{a} \in k$,
$\hat{a}$ is of the form~$\alpha(x)\beta(t)$ where $\alpha(x) \in
\mathbb{C}(x)$ and $\beta(t) \in \overline{\mathbb{C}(t)}$. Then
$\hat{b} \in \overline{\mathbb{C}(t)}[x]$. Suppose that
$\hat{b}=c_nx^n+c_{n-1}x^{n-1}+\cdots+c_0$ where $c_i
\in\overline{\mathbb{C}(t)}$ and $c_n \neq 0$. Then
$$
\sigma^m(\hat{b})-\hat{b}=nm c_nx^{n-1}+\cdots=\frac{\delta(\hat{a})}{\hat{a}}=\frac{\delta(\beta(t))}{\beta(t)}.
$$
So $n=1$ and $\hat{b}=\frac{\delta(\beta(t))}{m \beta(t)}x+c_0$.
\end{proof}

\begin{theorem}\label{thm4}
If~$\{\sigma(Y)=AY, \delta(Y)=BY\}$ with $A \in \mat_n(k_0)$ and~$ B
\in \gl_n(k_0)$  is irreducible over~$k_0$ and its Galois group
over~$k_0$ is solvable by finite, then there exists~$\ell\in \bZp$
with $\ell |n$ such that the system
$$
\sigma^{\ell}(Y)=A_\ell Y,
\quad \delta(Y)=BY
$$
is equivalent over $k$ to
\begin{equation}
   \begin{cases}
    \sigma^\ell(Y)=\diag(\Lambda(x)\beta_1(t),\Lambda(x)\beta_2(t),\cdots,\Lambda(x)\beta_m(t))Y,\medskip\\
    \delta(Y)=\diag\left(\frac{\delta(\beta_1(t))}{\ell\beta_1(t)}xI_\ell
    +C_1,\cdots,\frac{\delta(\beta_m(t))}{\ell\beta_m(t)}xI_\ell
    +C_m)\right)Y
   \end{cases}\label{normaleqns}
  \end{equation}
where $\Lambda(x)=\diag(\alpha(x),\cdots,\alpha(x+\ell-1)),
C_1=\diag(c_1,\cdots,c_\ell)$ and $m\ell{=}n$. Moreover,
$\alpha(x)\in \mathbb{C}(x)$ is standard with respect to
$\sigma^\ell$, $\beta_i(t),c_i \in {\overline{\mathbb{C}(t)}}$,
{and there exists~$g_i$ in the Galois group of the
original system 
over $k_0$ such
that~$\beta_i(t)=g_i(\beta_1(t))$ and~$C_i=g_i(C_1)$.}
\end{theorem}
\begin{proof}
By Theorem \ref{thm6}, it suffices to  prove the theorem for a
factor of the given system. Let $\{\sigma(Y)=\mathcal{A}Y,
\delta(Y)=\mathcal{B}Y\}$ be such a factor  with~$\mathcal{A}\in
\mat_\ell(k)$ and~$\mathcal{B}\in \gl_\ell(k)$. By Proposition
\ref{prop41},
 $\{\sigma^\ell(Y)=\mathcal{A}_\ell Y, \delta(Y)=\mathcal{B}Y\}$ is equivalent over
$k$ to
$$
  \sigma^{\ell}(Y)=\D Y,\quad
   \delta(Y)=\diag(b_1,\cdots,b_\ell)Y
$$
where $\D$ is as in Corollary~\ref{cor12} and $b_i \in k$
for~$i=1,\dots, \ell$. Since $\sigma^\ell$ and $\delta$ commute, we
have~$ \sigma^\ell(b_1)-b_1=\frac{\delta(a)}{a}.$
 By
Proposition \ref{prop6}, we have
$$
a=\frac{\sigma^\ell(f)}{f}\alpha(x)\beta_1(t)\quad\mbox{and}\quad
b_1=\frac{\delta(f)}{f}
+\frac{\delta(\beta_1(t))}{\ell\beta_1(t)}x+c_1
$$
where $\alpha(x) \in \mathbb{C}(x)$ is standard with respect to
$\sigma^\ell$, $c_1,\beta_1(t) \in \overline{\mathbb{C}(t)}$ and $f
\in k$. Then for $i=1,\cdots,\ell$,
$$
\sigma^{i-1}(a)=\frac{\sigma(\sigma^{i-1}(f))}{\sigma^{i-1}(f)}\alpha(x+i-1)\beta_1(t)
\,\,\mbox{ and }\, \,
b_i=\frac{\delta(\sigma_{i-1}(f))}{\sigma^{i-1}(f)}+\frac{\delta(\beta_1(t))}{\ell\beta_1(t)}x+c_i.
$$
 Let
$F=\diag(f,\sigma(f),\cdots,\sigma^{\ell-1}(f))$. Then the system
$$
\sigma^{\ell}(Y)=\D Y,\quad
   \delta(Y)=\diag(b_1,\cdots,b_\ell)Y
$$
 is equivalent over $k$ to
$$
   \sigma^\ell(Y)=\Lambda(x)\beta_1(t)Y, \quad
   \delta(Y)=\left(\frac{\delta(\beta_1(t))}{\ell\beta_1(t)}xI_\ell+C_1\right)Y
$$
under the transformation $Y\rightarrow FY$ where
$\Lambda(x)=\diag(\alpha(x),\cdots,\alpha(x+\ell-1))$ and
$C_1=\diag(c_1,\cdots,c_\ell).$ 
\end{proof}

\subsubsection{Normal Forms for {Systems} 
 of Prime Order}
{If a difference-differential system is of prime order~$n$, then the
integer~$\ell$ in Theorem \ref{thm4} equals either~$1$ or~$n$. For
the case where the system is reducible over~$k$, we can refine
Theorem~\ref{thm4} further in the following}
\begin{prop}\label{cor4}
{Assume that~$n$ is a prime number. Suppose that the
system~$
 \{  \sigma(Y)=AY, \delta(Y)=BY\}$   with~$A \in \mat_n(k_0)$ and~$
B \in \gl_n(k_0)$  is irreducible over $k_0$ and reducible over~$k$
and that its Galois group of is solvable by finite.  Then the system
is equivalent over~$k$ to}
\begin{equation}\label{norm1}
\begin{cases}
    \sigma(Y)=\alpha(x)\diag(\beta_1(t),\beta_2(t),\cdots,\beta_n(t))Y,\medskip \\
    \delta(Y)=\diag\left(\frac{\delta(\beta_1(t))}{\beta_1(t)}x
    +c_1,\cdots,\frac{\delta(\beta_n(t))}{\beta_n(t)}x
    +c_n)\right)Y
   \end{cases}
  \end{equation}
where $\alpha(x) \in \mathbb{C}(x)$ is standard with respect to
$\sigma$,
$ \beta_i(t)=g_i(\beta_1(t))\in {\overline{\mathbb{C}(t)}} $ and~$
c_i=g_i(c_1)\in {\overline{\mathbb{C}(t)}} $ for some~$g_i$ in the
Galois group of the original system over $k_0$.
\end{prop}
{Before discussing the other case where a
difference-differential system is irreducible over $k$, let us look
at the following}
\begin{lemma}\label{lem9}
Assume that $\sigma(Y)=AY$ with $A \in \mat_n(k_0)$ is equivalent
over $k$ to $\sigma(Y)=\bar{A}Y$ where
\[\begin{array}{cccc}
\bar{A} =&
\begin{pmatrix}
  0 & 1 & 0 & \cdots & 0 \\
  0 & 0 & 1 & \cdots & 0 \\
  \vdots & \vdots & \vdots & \vdots & \vdots \\
  0 & 0 & 0 & \cdots & 1 \\
  \beta(t)\alpha(x) & 0 & 0 & \cdots & 0
\end{pmatrix},
\end{array}\]
with~$\alpha(x) \in \mathbb{C}(x)$ and $\beta(t) \in
\overline{\mathbb{C}(t)}$. Then $\beta(t) \in \mathbb{C}(t).$
\end{lemma}

\begin{proof}
There exists $G \in GL_n(k)$ such that $\sigma(G)\bar{A}G^{-1}=A$.
Then
$$
  \det(\sigma(G))\det(\bar{A})\det(G^{-1})=\det(A).
$$
Since $\det(\sigma(G))=\sigma(\det(G))$ and
$\det(G^{-1})=\frac{1}{\det(G)}$, we have
$$(-1)^{n-1}\beta(t)\alpha(x)\frac{\sigma(\det(G))}{\det(G)}=\det(A).$$
Expand the rational functions in $x$ in the above equation as series
at $x=\infty$. Since
$\frac{\sigma(\det(G))}{\det(G)}=1+\frac{1}{x}Q$ where $Q \in
\overline{\mathbb{C}(t)}[[\frac{1}{x}]]$, one sees that $\beta(t)
\in \mathbb{C}(t)$.
\end{proof}

\begin{prop}  \label{nprop3}
{Let~$A,\bar{A} \in \mat_n(k_0)$. If $\sigma(Y)=AY$ and
$\sigma(Y)=\bar{A}Y$ are equivalent over $k$ then they are
equivalent over $k_0$}.
\end{prop}
\begin{proof}
Suppose that there exists $G \in \mat_n(k)$ such that
$\sigma(G)A=\bar{A}G.$ Then there exists $\gamma(t) \in
\overline{\mathbb{C}(t)}$ such that $G \in \mat_n(k_0(\gamma(t)))$.
Let $m=[k_0(\gamma(t)):k_0]$. Since
$1,\gamma(t),\cdots,\gamma(t)^{m-1}$ is a basis of $k_0(\gamma(t))$
over $k_0$,  we can  write
$$G=G_0+G_1\gamma(t)+\cdots+G_{m-1}\gamma(t)^{m-1}$$
where $G_i \in \gl_n(k_0)$. From $\sigma(G)A=\bar{A}G$, it follows
that~$\sigma(G_i)A=\bar{A}G_i$ for~$i=0,\cdots,m-1$. Let $\lambda$
be a parameter satisfying~$\sigma(\lambda)=\lambda$ and let $
H(\lambda)=\sum_{i=0}^{m-1}\lambda^iG_i. $ Therefore,
$\sigma(H(\lambda))A=\bar{A}H(\lambda).$ Since
$\det(G)=\det(H(\gamma(t)))\neq 0$, $\det(H(\lambda))$ is a nonzero
polynomial with coefficients in $k_0$. Hence there exists~$c \in
\mathbb{C}(t)$ such that $\det(H(c))\neq 0$.
So~$\sigma(H(c))A=\bar{A}H(c)$ and $H(c) \in \mat_n(k_0)$.
\end{proof}
{We now turn to the case where a difference-differential
system over~$k_0$ is irreducible over~$k$.}

\begin{prop}\label{prop51}
Suppose that~$ \{  \sigma(Y)=AY,  \delta(Y)=BY\}$ with $A \in
\mat_n(k_0)$ and~$ B \in \gl_n(k_0)$ is irreducible over $k$ and
that its Galois group over $k_0$ is solvable by finite. Then  the
system is equivalent over $k_0$ to
\begin{equation*}
\sigma(Y)=\bar AY, \quad \delta(Y)=\bar{B}Y
\end{equation*}
where~$\bar{B} \in \gl_n(k_0)$ and
$$
\bar A=\begin{pmatrix}
  0 & 1 & 0 & \cdots & 0 \\
  0 & 0 & 1 & \cdots & 0 \\
  \vdots & \vdots & \vdots & \vdots & \vdots \\
  0 & 0 & 0 & \cdots & 1 \\
  \beta(t)\alpha(x) & 0 & 0 & \cdots & 0
\end{pmatrix} \in \mat_n(k_0)
$$
with $\alpha(x) \in \mathbb{C}(x)$  standard with respect to
$\sigma^n$ and $ \beta(t) \in \mathbb{C}(t)$.
Moreover,~$\frac{\alpha(x+1)}{\alpha(x)}\neq \frac{\sigma^n(b)}{b}$
for any $b \in \CX(x).$
\end{prop}
\begin{proof}
By Proposition \ref{prop2}, the given system is equivalent over $k$
to the system~$\{\sigma (Y) =\bar{\mathcal{A}}Y, \delta (Y) =
\bar{\mathcal{B}} Y\}$ where~$\bar{\mathcal{B}} \in \gl_n(k)$ and
\[\begin{array}{cccc}
    \bar{\mathcal{A}}=&\begin{pmatrix}
         0 & 1 & 0 & \cdots & 0 \\
         0 & 0 & 1 & \cdots & 0 \\
         \vdots & \vdots & \vdots & \vdots & \vdots \\
         0 & 0 & 0 & \cdots & 1 \\
         a & 0 & 0 & \cdots & 0
      \end{pmatrix}
      \end{array}\in \mat_n(k)
\]
for some~$a \in k$. Since $\sigma$ and $\delta$ commute, we have~$
\sigma(\bar{\mathcal{B}})\bar{\mathcal{A}}=\delta(\bar{\mathcal{A}})+\bar{\mathcal{A}}\bar{\mathcal{B}}.
$
Let~$\bar{\mathcal{B}}=(\bar{b}_{ij})_{n\times n}$ where
$\bar{b}_{i,j} \in k.$ Then
$$
   \sigma(\bar{b}_{nn})-\bar{b}_{11}=\frac{\delta(a)}{a}, \quad
   \bar{b}_{nn}=\sigma(\bar{b}_{n-1,n-1}),\quad \cdots,\quad
   \bar{b}_{22}=\sigma(\bar{b}_{11}).
$$
Hence $\sigma^n(\bar{b}_{11})-\bar{b}_{11}=\frac{\delta(a)}{a}.$ By
Proposition \ref{prop6}, we
have~$a=\frac{\sigma^n(f)}{f}\alpha(x)\beta(t)$ with~$f \in k$,
$\alpha(x)\in \mathbb{C}(x)$  standard with respect to $\sigma^n$
and $\beta(t) \in \overline{\mathbb{C}(t)}.$ Then 
$\{\sigma (Y) =\bar{\mathcal{A}}Y, \delta (Y) =
\bar{\mathcal{B}} Y\}$ is equivalent over $k$ to 
$\{\sigma(Y)=\bar AY, \delta(Y)=\bar{B}Y\}$ under the transformation
$ Y\rightarrow \diag(f,\sigma(f),\cdots,\sigma^{n-1}(f))Y, $ where
$\bar{B} \in \gl_n(k_0)$ and
$$
\bar A=\begin{pmatrix}
  0 & 1 & 0 & \cdots & 0 \\
  0 & 0 & 1 & \cdots & 0 \\
  \vdots & \vdots & \vdots & \vdots & \vdots \\
  0 & 0 & 0 & \cdots & 1 \\
  \beta(t)\alpha(x) & 0 & 0 & \cdots & 0
\end{pmatrix}
$$
with $\alpha(x) \in \mathbb{C}(x)$ and~$\beta(t) \in
\overline{\mathbb{C}(t)}$. By Lemma \ref{lem9} and Proposition
\ref{nprop3}, the original system is equivalent over $k_0$
to~$\{\sigma(Y)=\bar AY, \delta(Y)=\bar{B}Y\}$ with~$\beta(t)\in
\CX(t)$. Assume
that~$\frac{\alpha(x+1)}{\alpha(x)}=\frac{\sigma^n(b)}{b}$ for some
$b \in \CX(x)$ and let $u = \sigma^{n-1}(b) \cdots \sigma(b) b$. We
have~$\frac{\alpha(x+1)}{\alpha(x)} = \frac{u(x+1)}{u(x)}$ thus
$\alpha(x) = cu(x)$ for some constant $c$ with respect to~$\sigma$.
Therefore~$c \in \CX$ since~$\alpha(x)$ and~$u(x)$ are both in~$
\CX(x)$. Let $P\in \mat_n(\overline{\CX(t)})$  be such that
\[\begin{array}{ccc}
    P^{-1}\begin{pmatrix}
  0 & 1 & 0 & \cdots & 0 \\
  0 & 0 & 1 & \cdots & 0 \\
  \vdots & \vdots & \vdots & \vdots & \vdots \\
  0 & 0 & 0 & \cdots & 1 \\
  c\beta(t) & 0 & 0 & \cdots & 0
\end{pmatrix}P=
   \diag(\tilde{\beta_1}(t),\tilde{\beta_2}(t),\cdots,\tilde{\beta_n}(t))
   \end{array}
\]
where the $\tilde{\beta_i}(t)$'s are the roots of $Y^n-c\beta(t)$.
Let
$$
F=\diag(1,b,\sigma(b)b,\cdots,\sigma^{n-2}(b)\cdots\sigma(b)b)P
$$
and~$\tilde{B}=F^{-1}\bar{B}F-F^{-1}\delta(F)$.
Then~$\{\sigma(Y)=\bar AY,  \delta(Y)=\bar{B}Y\}$
 is equivalent
over $k$ to
$$
\sigma(Y)=b\cdot
\diag(\tilde{\beta_1}(t),\tilde{\beta_2}(t),\cdots,\tilde{\beta_n}(t))Y,
   \quad
   \delta(Y)=\tilde{B}Y
$$
under the transformation $Y\rightarrow FY.$
Assume~$\tilde{B}=(\tilde{b}_{ij})_{n\times n}.$ Since $\sigma$ and
$\delta$ commute,
$$
\sigma(\tilde{b}_{ij})-\frac{\tilde{\beta_i}(t)}{\tilde{\beta_j}(t)}\tilde{b}_{ij}=0,
$$
for all~$i, j$ with~$1 \leq i \neq j \leq n$. Hence
$\tilde{b}_{ij}=0$ if $i \neq j$. In  other word, $\tilde{B}$ is of
diagonal form. This contradicts to the irreducibility over $k$ of
the original system. 
\end{proof}

\begin{lemma}\label{lem13}
Let $a \in k_0 \setminus \{0\}$,  $n$ be a positive integer
and~$m>0$ be the least integer such that
$\frac{\sigma^m(a)}{a}=\frac{\sigma^n(b)}{b}$ for some $b \in k_0$.
Then~$m|n$.
\end{lemma}
\begin{proof}
Suppose  that $\frac{\sigma^m(a)}{a}=\frac{\sigma^n(b)}{b}$ with $b
\in k_0$. Then for each $\ell >0$,
$$
\frac{\sigma^{\ell m}(a)}{a}=\frac{\sigma^n(c_\ell)}{c_\ell}
\quad\mbox{with $c_\ell \in k_0.$}
$$
 Let $n=\ell_1 m +\ell_2$
where $0 \leq \ell_2 \leq m-1$. Then
$$
     \frac{\sigma^{\ell_2}(a)}{a}=\frac{\sigma^{\ell_1
     m+\ell_2}(a)}{a}\frac{\sigma^{\ell_2}(a)}{\sigma^{\ell_1
     m+\ell_2}(a)}=\frac{\sigma^n(a)}{a}\frac{c}{\sigma^n(c)}
$$
  for some $c \in k_0$. Hence $\ell_2=0$ and so $m |n$.
\end{proof}

\begin{prop} \label{thm52}
Assume that $n$ is a prime number, the system
$$
 \sigma(Y)=AY, \quad \delta(Y)=BY
$$ with $A \in \mat_n(k_0)$ and~$ B
\in \gl_n(k_0)$ is irreducible over $k$ and its Galois group over
$k_0$ is solvable by finite. Then
$\{\sigma^n(Y)=A_nY,\delta(Y)=BY\}$ is equivalent over $k_0$ to
\[\begin{cases}
        \sigma^n(Y)=\beta(t)\diag(\alpha(x),\cdots,\alpha(x+n-1))Y\medskip \\
        \delta(Y)=\left(\frac{\delta(\beta(t))}{n\beta(t)}xI_n+\diag({\hat b}_1,\cdots,\hat{b}_n)\right)Y
      \end{cases}
\]
where $\alpha(x)$ and $\beta(t)$ are as in Proposition~\ref{prop51}
and $\hat{b}_i \in
      \mathbb{C}(t)$ for~$i=1, \dots, n$.
\end{prop}
\begin{proof}
By Proposition \ref{prop51}, $\{\sigma^n(Y)=A_nY, \delta(Y)=BY\}$ is
equivalent over~$k_0$ to the system
\begin{equation*}
\sigma^n(Y)=\beta(t)\cdot \diag(\alpha(x),\cdots,\alpha(x+n-1))\, Y,
\quad \delta(Y)=\bar{B}Y
\end{equation*}
with $\alpha(x)$ and $\beta(t)$ as in Proposition~\ref{prop51} and
$\bar{B} \in \gl_n(k_0).$ Let~$\bar{B}{=}(\bar{b}_{ij})_{n\times
n}$. From $\sigma^n\delta=\delta\sigma^n$, we have
\[\begin{cases}
     \sigma^n(\bar{b}_{ii})-\bar{b}_{ii}=\frac{\delta(\beta(t))}{\beta(t)},
     & i=1,\cdots,n, \medskip \\
     \sigma^n(\bar{b}_{ij})-\frac{\alpha(x+i)}{\alpha(x+j)}\bar{b}_{ij}=0,
     & 1 \leq i \neq j \leq n.
\end{cases}\]
Hence $\bar{b}_{ii}=\frac{\delta(\beta(t))}{n\beta(t)}x+\hat{b}_i$
with $\hat{b}_i \in \mathbb{C}(t)$. Note that~$n$ is prime and
$\frac{\alpha(x+1)}{\alpha(x)} {\neq} \frac{\sigma^n(b)}{b}$ for any
$b \in \mathbb{C}(x)$. Then by Lemma \ref{lem13},
$\frac{\alpha(x+i)}{\alpha(x)} \neq \frac{\sigma^n(b)}{b}$ for
any~$1 \leq i \leq n-1$ and  $b \in \mathbb{C}(x).$ Hence
$\bar{b}_{ij}=0$ for $i \neq j$. This concludes the {proposition}.
 \end{proof}

\subsection{{A} Decision Procedure for {Systems} of Prime Order}\label{sec6}
Consider a system~$ \{ \sigma(Y)=AY,\delta(Y)=BY\}$ over~$k_0$.
 Assume that the
order~$n$ is prime, the system is irreducible over~$k_0$ and its
Galois group is solvable by finite (or, equivalently, the system has
liouvillian solutions). By Proposition~\ref{cor4} and
Proposition~\ref{thm52}, either the original system has
hypergeometric solutions over $k$ or the system~$\{\sigma^n(Y)=AY,
\delta(Y)=BY\}$ has solutions which are the interlacing of
hypergeometric solutions over $k_0$. In this section, we will give a
decision procedure to find solutions of systems of both forms when
the order~$n$ is prime.
%
Our procedure relies on the following three facts {in the ordinary
cases}:
\begin{itemize}
  \item [$(A_1)$]
    we can compute all rational solutions in $k^n$ of an ordinary difference equation~$\sigma(Y)=AY$
    where $A \in \mat_n(k)$; (\cite{abramov-barkatou, abramov2,abramov3,hoeij1});
  \item [$(A_2)$]
    we can compute all hypergeometric solutions over $\mathbb{C}(x)$
    of an ordinary  difference equation~$\sigma(Y)=\hat{A}Y$
    where $\hat{A} \in \mat_n(\mathbb{C}(x))$ (\cite{hendrikssinger,labahn-li,ziming-etal,
     petkovsek,petkovsek-salvy,hoeij2});
  \item [$(A_3)$]
    we can compute all hyperexponential solutions over
    $\overline{\mathbb{C}(t)}$ of an ordinary  differential equation~$\delta(Y)=\hat{B}Y$ where $\hat{B} \in
    \mat_n(\mathbb{C}(t))$ (\cite{kovacic,labahn-li, ziming-etal,singer1,hoeij-etal}).
\end{itemize}
In the following subsections, we will reduce the problem of finding
solutions of $\{ \sigma(Y)=AY, \delta(Y)=BY\}$ or of
$\{\sigma^n(Y)=AY, \delta(Y)=BY\}$ to that {in the
ordinary cases as indicated above.
We have two case distinctions  according to the reducibility of~$\{
\sigma(Y)=AY, \delta(Y)=BY\}$ over  $k$}.
\subsubsection{The Decision Procedure for the Reducible Case}\label{sec6.1}
Assume  that~$\{ \sigma(Y)=AY, \delta(Y)=BY\}$  is reducible over
$k$. Proposition~\ref{cor4} implies that this system has
hypergeometric solutions of the form $W_ih_i$ for~$i=1,\dots, n$,
where~$W_i \in k^n$ and~$h_i$ satisfies
$$\sigma(h_i)=\alpha(x)\beta_i(t)h_i \quad\mbox{and}\quad \delta(h_i)
=\left(\frac{\delta(\beta_i(t))}{\beta_i(t)}x+c_i\right)h_i
$$
{with $\alpha(x) \in \mathbb{C}(x)$  standard with
respect to $\sigma$,
$ \beta_i(t)=g_i(\beta_1(t))\in {\overline{\mathbb{C}(t)}}$
and~$c_i=g_i(c_1)\in {\overline{\mathbb{C}(t)}} $ for some~$g_i$ in
the Galois group of the original system over $k_0$.}
Substituting each $W_ih$ into the original
system, 
we get
\begin{equation}\label{eqn61}
  \sigma(W_i)=\frac{A}{\alpha(x)\beta_i(t)}W_i \quad\mbox{and}\quad
  \delta(W_i)=\left(B-\frac{\delta(\beta_i(t))}{\beta_i(t)}x-c_i\right)W_i.
\end{equation}
So,  to compute hypergeometric solutions of~$\{ \sigma(Y)=AY,
\delta(Y)=BY\}$ it suffices to find~$\alpha(x), \beta_i(t), c_i$ and
$W_i$ satisfying (\ref{eqn61}).

\begin{remark}\label{rem6.1}
{The equalities~\eqref{eqn61} still hold when  replacing
$\alpha(x)$ by $\frac{\sigma(g)\alpha(x)}{g}$ and~$W_i$ by
$\frac{W_i}{g}$ for $g \in \mathbb{C}(x)$.}  So in the sequel, we
will compute a suitable~$\frac{\sigma(g)\alpha(x)}{g}$  instead of
$\alpha(x)$.
\end{remark}
\vspace{.2in}{\bf \underline{Computing $\alpha(x)$:}}
 By Proposition \ref{cor4}, there exists $G \in
\mat_n(k)$ such that
$$
\frac{\sigma(\det(G))}{\det(G)}\alpha(x)^n\prod_{i=1}^n\beta_i(t)=\det(A).
$$
Without loss of generality, we  assume that the numerator and
denominator of $\alpha(x)$ are monic. Expanding the functions in the
above equality as series at $x=\infty$, one can compute
$\prod_{i=1}^n\beta_i(t)$ from the series expansion of $\det(A)$ at
$x=\infty$. Let $\tilde{a}=\frac{\det(A)}{\prod_{i=1}^n\beta_i(t)}.$
Rewrite
~$\tilde{a}=\frac{\sigma(b)}{b}\bar{a}$ where $b,\bar{a} \in k_0$
and $\bar{a}$ is standard with respect to $\sigma$. Then
\begin{equation}\label{eqn64}
\bar{a}=\frac{\sigma(g)}{g}\alpha(x)^n \quad\mbox{for some $g \in
k_0.$}
\end{equation}
 From Proposition \ref{cor4}, $\alpha(x)$ is
standard with respect to~$\sigma$ and so is $\alpha(x)^n$.
Proposition~\ref{prop:1} below shows that $\frac{\sigma(g)}{g} \in
\mathbb{C}(x)$ and thus~$\bar{a} \in \mathbb{C}(x)$.
Moreover,~$\bar{a}$ has the form
$\left(\frac{\sigma(\bar{g})}{\bar{g}}\alpha(x)\right)^n$ for some
$\bar{g}\in \mathbb{C}(x)$. {To prove
Proposition~\ref{prop:1}, let us introduce a  notation used in
\cite[Section 2.1]{putsinger1}.}
\begin{define}\label{def3}
A divisor $D$ on $\mathbb{P}^1(\overline{\mathbb{C}(t)})$ is defined
to be  a finite formal expression $\sum n_p[p]$ {with~$p \in
\mathbb{P}^1(\overline{\mathbb{C}(t)})$ and  $n_p \in \mathbb{Z}$.}
The support of a divisor~$D$,  denoted  $\supp(D)$, is the finite
set of all $p$ with $n_p \neq 0$. Let~$p \in \supp(D)$. The
$\mathbb{Z}$-orbit $E$ of $p$ in $\supp(D)$ is defined to be
$$
    E(p, \supp(D))=\{p+i | i \in \mathbb{Z}\,\,\mbox{and}\,\,p+i \in \supp(D)\}.
$$
\end{define}
As usual, the divisor $\divs(f)$ of a rational function $f\in
k\setminus\{0\}$ is given by~$\divs(f)=\sum \ord_p(f)[p]$, where
$\ord_p(f)$ denotes the order of $f$ at the point~$p$. It is clear
that $\divs(fg)=\divs(f)+\divs(g)$. Moreover, if~$p$ is
in~$\supp(\divs(f))$ but not in $\supp(\divs(fg))$, then $p \in
\supp(\divs(f))\cap \supp(\divs(g))$. By Definition~\ref{def3}, if
$f \in k\setminus \{0\}$ is standard with respect to $\sigma$, then
$E(p,\supp(\divs(f)))=\{p\}$ for each $p \in \supp(\divs(f))$.
\begin{prop}\label{prop:1}
Assume that $f,g\in k\setminus\{0\}$ and $f$ is standard with
respect to $\sigma$. If $\sigma(g)g^{-1}f$ is standard with respect
to $\sigma$, then
$$
   \frac{\sigma(g)}{g}=\prod_i \frac{(x+k_i-c_i)^{m_i}}{(x-c_i)^{m_i}}
$$
with~$k_i\in \mathbb{Z}$,~$m_i\in \bZp$,~$ c_i \in
\overline{\mathbb{C}(t)}$ and $\disp(\prod_i(x-c_i))=0.$ Moreover,
for each $i$, either~$\ord_{c_i}(f)=m_i$ or
$\ord_{c_i-k_i}(f)=-m_i$.
\end{prop}
\begin{proof}
Let $H=\sigma(g)g^{-1}f$, $S_1=\supp(\divs(f))$,
$S_2{=}\supp(\divs(\sigma(g)g^{-1}))$ and~$S_3=\supp(\divs(H))$. By
Lemma 2.1 in \cite{putsinger1},
$$
\sum_{q \in E(p,S_2)}\ord_q\left(\frac{\sigma(g)}{g}\right)=0
\quad\mbox{for each $p \in S_2$.}
$$
Then $|E(p,S_2)| \geq 2$ for each $p \in \supp(S_2)$. Since $H$ and
$f$ are standard,
$$
|E(p,S_2)\cap S_3|\leq 1\,\, \mbox{and} \,\, |E(p,S_2)\cap S_1|
\leq 1
$$
thus $|E(p,S_2)\cap (S_1\cup S_3)|\leq 2.$ From~$S_2 \subseteq
S_1\cup S_3,$ we have~$|E(p,S_2)|\leq 2.$ Hence  for each $p \in
S_2,$
$$|E(p,S_2)|=2, \,\,|E(p,S_2)\cap S_1|=1
\,\,\mbox{and}\,\,|E(p,S_2)\cap S_3|=1.
$$
From $|E(p,S_2)|=2$ and $|E(p,S_2)\cap S_3|=1$,
either~$\ord_p(\sigma(g)g^{-1})=-\ord_p(f)$
or~$\ord_{p+j_0}(\sigma(g)g^{-1})=-\ord_{p+j_0}(f)$ with $p+j_0 \in
E(p,S_2).$ The proposition holds.
\end{proof}
Let $g$ be as in (\ref{eqn64}). Since $\alpha(x) \in \mathbb{C}(x)$,
 $g$ can be chosen  in $\mathbb{C}(x)$ according to Proposition~\ref{prop:1}.
Then $\bar{a} \in \mathbb{C}(x)$. Moreover,
$$\frac{\sigma(g)}{g}=\prod_i
\frac{(x+k_i-c_i)^{m_i}}{(x-c_i)^{m_i}}
$$
where $m_i$ has the form~$\bar{m}_i n$ for some~$\bar{m}_i\in \bZp$
since~$m_i$ is either~$\ord_{c_i}(\alpha(x)^n)$ or
$-\ord_{c_i-k_i}(\alpha(x)^n)$. Let $
\bar{g}=\prod_i\prod_{j=0}^{k_i-1}(x+j-c_i)^{\bar{m}_i}. $ Then
$$
\left(\frac{\sigma(\bar{g})}{\bar{g}}\right)^n=\frac{\sigma(g)}{g}\quad
\mbox{and} \quad
\bar{a}=\left(\frac{\sigma(\bar{g})}{\bar{g}}\alpha(x)\right)^n.
$$
Note that the numerator and the denominator of $\alpha(x)$ are
monic, so we can compute
$\frac{\sigma(\bar{g})}{\bar{g}}\alpha(x)$ from $\bar{a}$.
\begin{example} \label{example1}Consider the
integrable system
$$
\sigma(Y)=AY, \quad\delta(Y)=BY
$$
where
\begin{align*}
   &\begin{array}{ccc}
      A=
         \begin{pmatrix}
           -\frac{t(x^2+1)(t^2+1-x)}{t^2-x-1}&
           -\frac{x^2+1}{t^2-x-1}\\[0.1in]
          \frac{(x^2+1)(t^4+t^2-x^2-x)}{t^2-x-1}&
          \frac{t(x^2+1)(t^2-x)}{t^2-x-1}
         \end{pmatrix},
    \end{array} \\[0.1in]
   &\begin{array}{cc}
      B=\begin{pmatrix}
           -\frac{-2xt^3-t^4-t^2+t^5+3t^3+2t+x^2t+t^2x+x}{(t^2-x)(t^2+1)}
           &-\frac{t^2}{(t^2-x)(t^2+1)}\\[0.1in]
           \frac{-t^2x^2+t^6+2t^4+t^2-x^2+2t^2x+x}{(t^2-x)(t^2+1)}&
        \frac{-x^2t-t^2x-x+t^5+2t^3-xt+t^4+t^2}{(t^2-x)(t^2+1)}
      \end{pmatrix}.
   \end{array}
\end{align*}
We have
$$
\det(A)=-\frac{(x^2+1)^2(t^4-t^2x+t^2-x)}{t^2-x-1}=-(t^2+1)x^4+(t^2+1)x^3+\cdots.
$$
Thus $\beta_1(t)\beta_2(t)=-(t^2+1).$ Let
$\tilde{a}=-\frac{\det(A)}{t^2+1}$ and  write
$$
   \tilde{a}=\frac{t^2-x}{t^2-(x+1)}(x^2+1)^2.
$$
Then $\alpha(x)=x^2+1.$
\end{example}

\vspace{.2in}\noindent{\bf \underline{Computing $\beta_i(t)$:}} We
first prove the following
\begin{lemma}\label{lem61}
Either~$ \beta_1(t)=\cdots=\beta_n(t) \in \mathbb{C}(t) $ or
$\beta_1(t),\cdots,\beta_n(t)$ are the conjugate roots of an
irreducible polynomial of degree $n$ with coefficients
in~$\mathbb{C}(t)$.
\end{lemma}
\begin{proof}
{Let  $R_0$  be a \sd PV extension of~$k_0$ for~$ \{ \sigma(Y)=AY,
\delta(Y){=}BY\}$ and~$P=\prod_{i=1}^n(X-\beta_i(t))$. From the
proof of Theorem \ref{thm6}, one sees that~$\gal(R_0/k_0)$ permutes
the $\beta_i(t)$.
Furthermore, the orbits of the $\beta_i(t)$ under this group action
all have the same size. Therefore, $P$ is a polynomial with
coefficients in $\mathbb{C}(t)$. Since $n$ is prime, either $P$ is
irreducible or all the factors of $P$ in $\mathbb{C}(t)[X]$ are of
degree one. This concludes the   lemma.}
\end{proof}
 The following two notions can be found in
\cite{barkatou-chen,barkatou}.
\begin{define}
Let $H=(h_{ij})_{n\times n} \in \mat_n(k_0)$. The {\rm order} of $H$
at $\infty$ is defined as
$$
\ord_{\infty}(H)=\min\{\ord_{\infty}(h_{ij})\}
$$
where $\ord_{\infty}(h_{ij})$ is the order of   $h_{ij}$ at
$\infty$.
\end{define}
We rewrite $H$ into the form
$$
   H=\left(\frac{1}{x}\right)^{\ord_{\infty}(H)}\left(H_0+H_1\frac{1}{x}+\cdots\right)
$$
where $H_i\in \gl_n(\mathbb{C}(t))$ and $H_0 \neq 0.$
\begin{define}
The rational number
$$
     m(H)=-\ord_{\infty}(H)+\frac{\rank(H_0)}{n}
$$
is called  {\rm the first Moser order} of $H$. And
$$
    \mu(H)=\min\{m(\sigma(G)HG^{-1})| G \in
    \mat_n(k)\}
$$
is called the {\rm Moser invariant} of $H$. A matrix $H$ is said to
be  irreducible if~$m(H)=\mu(H)$, otherwise it is called reducible.
\end{define}
Given 
$H \in \mat_n(k_0)$,  one can use the algorithm in
\cite{barkatou-chen,barkatou} to compute~$G \in \mat_n(k_0)$ such
that~$\tilde{H}:=\sigma(G)HG^{-1}$ is irreducible. So we can assume
that~$\frac{A}{\alpha(x)}$ is irreducible where $A$ and $\alpha(x)$
are as in (\ref{eqn61}). Under this assumption, we will show that
\begin{equation*}
 \frac{A}{\alpha(x)}
= \tilde{A}_0 + \tilde{A}_1 \frac{1}{x} + \ldots
\end{equation*}
with $\tilde{A}_i \in \gl_n(\mathbb{C}(t))$ for each~$i$ and that
all the $\beta_i(t)$'s are eigenvalues of $\tilde{A}_0$. The
following lemma can be deduced from the results of Barkatou in
\cite{barkatou}. We will present a self contained proof due to
Reinhart Shaefke.

\begin{lemma}\label{lem62}
Let $G \in \mat_n(k)$ and assume that
$\ord_{\infty}(\sigma(G^{-1})G)=0$. Then all the eigenvalues of
$\sigma(G^{-1})G|_{x=\infty}$ are 1.
\end{lemma}
\begin{proof}
Let $H=\sigma(G^{-1})G$.  Then~$H=H_0+H_1\frac{1}{x}+\cdots$
with~$H_i \in \gl_n(\overline{\mathbb{C}(t)})$  and $H_0 \neq 0$. We
now show that $H_0-I_n$ is nilpotent. For a positive integer~$m$,
consider a map~$L_m: \gl_n(k)\rightarrow \gl_n(k)$ given
by~$U\mapsto \sigma(U){-}\sigma^m(H)U$ for any $U\in \gl_n(k)$.
Set~$P_m=L_m\circ L_{m-1} \circ \cdots \circ L_0(I_n)$ where $\circ$
denotes the composition of  maps. Then~$
P_m|_{x=\infty}=(I_n-H_0)^{m+1}. $ On the other
hand,~$L_m(\sigma^m(G^{-1})V)=\sigma^{m+1}(G^{-1})\Delta(V)$ where
$\Delta=\sigma-\id$ is a difference operator and $V \in \gl_n(k)$.
Hence $P_m=\sigma^{m+1}(G^{-1})\Delta^{m+1}(G).$ Note that when $m$
increases,~$\ord_{\infty}(\Delta^{m+1}(G))$ increases but
$\ord_{\infty}(\sigma^{m+1}(G^{-1}))$ is invariant. Then for a
sufficiently large $m$, $P_m|_{x=\infty}=0$. This concludes the
lemma.
\end{proof}
Now we can prove the following
\begin{prop}\label{prop63}

$\ord_{\infty}\left(\frac{A}{\alpha(x)}\right)=0$ and~$\beta_1(t),
\dots, \beta_n(t)$ are eigenvalues
of~$\frac{A}{\alpha(x)}\mid_{x=\infty}$.
\end{prop}
\begin{proof}
 By Proposition \ref{cor4}, there exists $G \in \mat_n(k)$ such that
$$
\sigma(G)\frac{A}{\alpha(x)}G^{-1}=\diag(\beta_1(t),\cdots,\beta_n(t)).
$$
This implies that
$\ord_{\infty}\left(\det\left(\frac{A}{\alpha(x)}\right)\right)=0$
and
$m\left(\frac{A}{\alpha(x)}\right)=\mu\left(\frac{A}{\alpha(x)}\right)\leq
1.$ By  the property of orders,
$$
\ord_{\infty}\left(\frac{A}{\alpha(x)}\right)\leq \frac{1}{n}
\ord_{\infty}\left(\det\left(\frac{A}{\alpha(x)}\right)\right)=0.
$$
Since $m\left(\frac{A}{\alpha(x)}\right)\leq 1$,
$\ord_{\infty}\left(\frac{A}{\alpha(x)}\right)=0$ by the definition
of the first Moser orders. Therefore,
$$
\frac{A}{\alpha(x)} = \tilde{A}_0 + \tilde{A}_1 \frac{1}{x} + \ldots
$$
where $\tilde{A}_i \in\gl_n(\mathbb{C}(t))$ and $\tilde{A}_0 \neq
0$.
 From (\ref{eqn61}), $\sigma(Y)=\frac{A}{\alpha(x)\beta_i(t)}Y$
has a rational solution $W_i$ in~$ k^n$. Suppose that
$$
W_i=\left(\frac{1}{x}\right)^{\ord_{\infty}(W_i)}\left(W_{i0}+\frac{1}{x}W_{i1}+\cdots
\right)
$$
where $W_{ij} \in \overline{\mathbb{C}(t)}^n$ and $W_{i0}\neq 0$.
Then~$
    W_{i0}=\frac{\tilde{A}_0}{\beta_i(t)}W_{i0}.
$ Since~$W_{i0}\neq 0$, we
have~$\det\left(I_n-\frac{\tilde{A}_0}{\beta_i(t)}\right)=0$. Hence
all the $\beta_i(t)$ are the eigenvalues of $\tilde{A}_0$. If
the~$\beta_i(t)$ are the conjugate roots of some irreducible
polynomial with degree~$n$, then   they  are clearly eigenvalues of
$\tilde{A}_0$. Thus  by Lemma \ref{lem61} we only need to consider
the case $\beta_1(t)=\cdots=\beta_n(t) \in \mathbb{C}(x)$. In this
case,~$\frac{A}{\alpha(x)}=\beta_1(t)\sigma(G^{-1})G$. Since
$\ord_{\infty}\left(\frac{A}{\alpha(x)}\right)=0$, we
have~$\ord_{\infty}(\sigma(G^{-1})G)=0.$ By Lemma \ref{lem62}, all
the eigenvalues of~$\sigma(G^{-1})G|_{x=\infty}$ equal 1. Hence all
the eigenvalues of $\tilde{A}_0$ equal $\beta_1(t)$.
\end{proof}
\begin{example}{\bf (Continued)} Let
$\bar{A}=\frac{A}{x^2+1}$. From the process in \cite{barkatou}, we
can find an irreducible matrix $\tilde{A}$ which is equivalent to
$\bar{A}$ where
\[
\begin{array}{cc}
   \tilde{A}=
    \begin{pmatrix}
        -\frac{(x+1)t(t^2+1-x)}{(t^2-x-1)x} & -\frac{x+1}{t^2-x-1}\medskip\\
        \frac{t^4+t^2-x^2-x}{(t^2-x-1)x}& \frac{t(t^2-x)}{t^2-x-1}
    \end{pmatrix}.
   \end{array}
\]
Write $\tilde{A}=\tilde{A}_0+\tilde{A}_1\frac{1}{x}+\cdots$ where
$$
    \tilde{A}_0=
    \begin{pmatrix}
       -t & 1 \\
       1 & t
    \end{pmatrix}\quad \mbox{and}\quad
    \tilde{A}_1=
    \begin{pmatrix}
       t & t^2 \\
       t^2 & -t
    \end{pmatrix}.
$$
The eigenvalues of $\tilde{A}_0$ are $\pm \sqrt{t^2+1}.$ So
$\beta_1(t){=}\sqrt{t^2+1}$ and~$\beta_2(t){=}-\sqrt{t^2+1}.$
\end{example}
\vspace{.1in} \noindent{\bf \underline{Computing $c_i$ and $W_i$:}}
\label{sec6.1.3} Let $\Lambda(t) =\diag(\beta_1(t), \ldots,
\beta_n(t))$. From $(A_1)$ we can find a matrix $G \in \mat_n(k)$
such that $\sigma(G)\alpha(x)\Lambda(t) = AG$. Let
$\bar{B}=G^{-1}BG-G^{-1}\delta(G).$ Then $\bar{B} \in \gl_n(k)$ and
the system~$  \{  \sigma(Y){=}AY,  \delta(Y){=}BY\}$ is equivalent
over $k$ to
\begin{equation}\label{eqn62}
  \sigma(Y)=\alpha(x)\Lambda(t)Y, \quad \delta(Y)=\bar{B}Y.
\end{equation}
Note that $G$ may not be the required transformation matrix in
Proposition~\ref{cor4}, so $\bar{B}$ may not be of diagonal form.
Since $\sigma\delta=\delta\sigma$, the same argument as in the proof
of Proposition \ref{prop51} implies the following conclusions:
\begin{itemize}
 \item [$(i)$]
    If  $\beta_i(t)\neq \beta_j(t)$ for all~$i,j$ with $1\leq i \neq j\leq n$,
    then
$$
\bar{B}=\diag\left(\frac{\delta(\beta_1(t))}{\beta_1(t)}x+c_1,\cdots,\frac{\delta(\beta_n(t))}{\beta_n(t)}x+c_n\right)
$$
with $c_i \in \overline{\mathbb{C}(t)}$;
  \item [$(ii)$]
     If $\beta_1(t)=\cdots=\beta_n(t) \in \mathbb{C}(t)$, then by
     Proposition \ref{nprop3}, $G$ can be chosen in $\mat_n(k_0)$.
     Thus
     $$
           \bar{B}=\hat{B}+\frac{\delta(\beta_1(t))}{\beta_1(t)}xI_n
     $$
     with $\hat{B} \in \gl_n(\mathbb{C}(t))$.
\end{itemize}
In the case (i), we obtain the~$c_i$'s, and the $W_i$'s are just the
columns of~$G$. For the case (ii), since (\ref{eqn62}) is equivalent
over $k$ to (\ref{norm1}), there exists~$\hat{G} \in \mat_n(k)$ such
that
\[\begin{cases}
 \sigma(\hat{G})=\hat{G},\\
\delta(\hat{G})+\hat{G}\left(\frac{\delta(\beta_1(t))}{\beta_1(t)}xI_n
+\diag(c_1,\cdots,c_n)\right)=\left(\hat{B}+\frac{\delta(\beta_1(t))}{\beta_1(t)}xI_n\right)\hat{G}.
\end{cases}\]
Hence $\hat{G} \in \mat_n(\overline{\mathbb{C}(t)})$ and
$\delta(Y)=\hat{B}Y$ is equivalent over $\overline{\mathbb{C}(t)}$
to
$$
\delta(Y)=\diag(c_1,\cdots,c_n)Y.
$$
Solving the system $\delta(Y)=\hat{B}Y$ by ($A_3$), we get the
$c_i$'s and the $W_i$ are just the columns of $\hat{G}G.$
\begin{example}{\bf
(Continued)}
   Let $\Lambda(t)=\diag(\sqrt{t^2+1},-\sqrt{t^2+1}).$ From~($A_1$), we can
   obtain
   $G \in \mat_2(k)$ such that $\sigma(G)(x^2+1)\Lambda(t)=AG$ where
   \[ \begin{array}{cc}
        G=
        \begin{pmatrix}
            \frac{t-\sqrt{t^2+1}}{2(t^2-x)}&
            \frac{t+\sqrt{t^2+1}}{2(t^2-x)}\medskip\\
            \frac{-x+t\sqrt{t^2+1}}{2(t^2-x)}&
            -\frac{x+t\sqrt{t^2+1}}{2(t^2-x)}
        \end{pmatrix}.
     \end{array}\]
Then
\begin{align*}
   \bar{B}&=G^{-1}BG-G^{-1}\delta(G)\medskip\\
   &=\begin{pmatrix}
      \frac{xt}{t^2+1}+\sqrt{t^2+1}+1 &
      0\\
      0 & \frac{xt}{t^2+1}-\sqrt{t^2+1}+1
   \end{pmatrix}.
\end{align*}
Hence  $c_1=\sqrt{t^2+1}+1$, $c_2=-\sqrt{t^2+1}+1$ and $W_i$ is the
$i$-th column of $G$ for~$i=1,2$. Furthermore, a basis of the
solution space is
\[ \begin{array}{cccc}
  h(\sqrt{t^2+1})^xe^{t+\int{\sqrt{t^2+1}}dt}
  \begin{pmatrix}
    \frac{t-\sqrt{t^2+1}}{2(t^2-x)}\\[0.1in]
    \frac{t\sqrt{t^2+1}-x}{2(t^2-x)}
  \end{pmatrix},
   h(-\sqrt{t^2+1})^xe^{t-\int{\sqrt{t^2+1}}dt}
  \begin{pmatrix}
    \frac{t+\sqrt{t^2+1}}{2(t^2-x)}\\[0.1in]
    \frac{x+t\sqrt{t^2+1}}{2(x-t^2)}
  \end{pmatrix}
\end{array}\]
where $h$ satisfies that $\sigma(h)=(x^2+1)h$ and $\delta(h)=0.$
\end{example}
\subsubsection{The Decision Procedure for the Irreducible Case}\label{sec3.2.2}
Assume that~$ \{  \sigma(Y)=AY,\delta(Y)=BY\}$ with $A \in
\mat_n(k_0)$ and~$ B \in \gl_n(k_0)$ is an irreducible system over
$k$ and its Galois group over $k_0$ is solvable by finite. By
Proposition \ref{thm52}, the system $\{\sigma^n(Y)=A_nY,
\delta(Y)=BY\}$ has solutions of the form $W_ih_i$ for~$
i=1,\cdots,n,$ where $W_i \in k_0^n$ and  $h_i$ satisfies
$$
\sigma^n(h_i)=\alpha(x+i-1)\beta(t)h_i,\quad \delta(h_i)
=\left(\frac{\delta(\beta(t))}{n\beta(t)}x+\hat{b}_i\right)h_i
$$
with $\alpha(x),\beta(t)$ and $\hat{b}_i$ as in
Proposition~\ref{thm52}. Substituting $Y=W_ih_i$
into~$\{\sigma^n(Y)=A_nY, \delta(Y)=BY\}$, we have
\begin{equation}\label{eqn3.2.2.2}
  \sigma^n(W_i)=\frac{A_n}{\beta(t)\alpha(x+i-1)}W_i\quad\mbox{and}\quad
  \delta(W_i)=\left(B-\frac{\delta(\beta(t))}{n\beta(t)}x-\hat{b}_i\right)W_i.
\end{equation}
To compute $W_ih_i$, it suffices to compute $\alpha(x), \beta(t),
W_i$ and~$\hat{b}_i$ which satisfy~(\ref{eqn3.2.2.2}).
 Without loss of generality, we assume that the numerator and
denominator of $\alpha(x)$ are monic. By Proposition \ref{prop51},
there exists $G \in \mat_n(k_0)$ such that
$$
    \frac{\sigma(\det(G))}{\det(G)}(-1)^{n-1}\alpha(x)\beta(t)=\det(A).
$$
Expanding $\det(A)$ as a series in $\frac{1}{x}$, we get that
$(-1)^{n-1}\beta(t)$ is the leading coefficient of the series. Hence
we can obtain $\beta(t)$ from $\det(A)$. In this case, we can not
find $\alpha(x)$ by the method used in Section~\ref{sec6.1}. However
we can reduce this problem to working with difference equations over
$\mathbb{C}(x)$. By Proposition \ref{thm52}, there exists~$G \in
\mat_n(k_0)$ (the same as that in Proposition~\ref{prop51}) such
that
$$
\sigma^n(G)\cdot\diag(\alpha(x),\cdots,\alpha(x+n-1))=\frac{A_n}{\beta(t)}G.
$$
Assume that $t=p$ is not a pole of the entries of
$\frac{A_n}{\beta(t)}$ and such that
$\det\left(\frac{A_n}{\beta(t)}|_{t=p}\right)\neq 0$. Let
$\frac{A_n}{\beta(t)}=\bar{A}_0+(t-p)\bar{A}_1+\cdots$ where
$\bar{A}_i \in \gl_n(\mathbb{C}(x))$.  We will show that
$\alpha(x)$ can be found by examining the hypergeometric solutions
of~$\sigma^n(Y)=\bar{A}_0Y$.  This will follow from the next
proposition.
\begin{prop}\label{prop3.2.2.1} Some factor of $\sigma^n(Y)=\bar{A}_0Y$ is
equivalent over $\CX(x)$ to some factor of
$\sigma^n(Y)=\diag(\alpha(x),\alpha(x+1),\cdots,\alpha(x+n-1))Y.$
\end{prop}
\begin{proof}
Let $G$ be as above and let
$\Psi(x)=\diag(\alpha(x),\cdots,\alpha(x+n-1))$. We may multiply $G$
by a power of $t-p$ and assume that
$G=\bar{G}_0+(t-p)\bar{G}_1+\cdots$ where $\bar{G}_0\neq 0$
and~$\bar{G}_i \in \gl_n(\mathbb{C}(x))$. Then
\begin{align*}
   \sigma^n(\bar{G}_0+(t-p)\bar{G}_1+\cdots)\Psi(x)
   =(\bar{A}_0+\cdots)
   (\bar{G}_0+(t-p)\bar{G}_1+\cdots).
\end{align*}
Therefore $\sigma^n(\bar{G}_0)\Psi(x)=\bar{A}_0\bar{G}_0$. Let
$r=\rank(\bar{G}_0).$ Then $r>0$ because~$\bar{G}_0 {\neq} 0.$ There
exist~$P \in \mat_n(\CX(x))$ and $Q$ which is  a product of some
permutation matrices such that
\[
  \begin{array}{ccc}
   \tilde{G}=P\bar{G}_0Q=\begin{pmatrix}
          0 & 0 \\
          \tilde{G}_{21} & \tilde{G}_{22}
    \end{pmatrix}
  \end{array}
\]
where $\tilde{G}_{22} \in \mat_r(\CX(x)).$ Then
\begin{equation}
\label{eqn3.2.2.1}
    \sigma^n(\tilde{G})\diag(\alpha(x+k_1),\cdots,\alpha(x+k_n))=\sigma^n(P)\bar{A}_0P^{-1}\tilde{G}
\end{equation}
where $k_1,\cdots,k_n$ are a permutation of $\{0,1,\cdots,n-1\}.$
Now let
\[
  \begin{array}{ccc}
   \tilde{A}=\sigma^n(P)\bar{A}_0P^{-1}=\begin{pmatrix}
          \tilde{A}_{11} & \tilde{A}_{12} \\
          \tilde{A}_{21} & \tilde{A}_{22}
    \end{pmatrix} \quad\mbox{  where
$\tilde{A}_{22}\in \gl_r(\CX(x))$,}
  \end{array}
\]
and $D_2=\diag(\alpha(x+k_{n-r+1}),\cdots,\alpha(x+k_n))$.
From~(\ref{eqn3.2.2.1}), we have~$\tilde{A}_{12}\tilde{G}_{22}{=}0$
and~$\sigma^n(\tilde{G}_{22})D_2=\tilde{A}_{22}\tilde{G}_{22}. $
Since $\tilde{G}_{22}\in \mat_r(\CX(x))$, we have
$\tilde{A}_{12}=0$. Therefore~$\sigma^n(Z)=\tilde{A}_{22}Z$ is a
factor of~$\sigma^n(Y)=\bar{A}_0Y$, which is equivalent
over~$\CX(x)$ to $\sigma^n(Z)=D_2Z.$
\end{proof}
\begin{remark}
For almost all of $p\in \CX$,
$\sigma^n(Y)=\frac{A_n}{\beta(t)}|_{t=p}Y$ is equivalent over
$\CX(x)$ to
$$
     \sigma^n(Y)=\diag(\alpha(x),\alpha(x+1),\cdots,\alpha(x+n-1))Y.
$$
since $G|_{t=p}$ is invertible.
\end{remark}
The same argument as in Remark \ref{rem6.1} implies that it is
enough to compute~$\frac{\sigma^n(g)\alpha(x)}{g}$ for some suitable
$g \in \mathbb{C}(x)$ instead of $\alpha(x)$. We can use
Proposition~\ref{prop3.2.2.1} to find
$\frac{\sigma^n(g)\alpha(x+k)}{g}$ with~$k\in \bZ$ and $g\in \CX(x)$
as follows. From Theorem 3 in \cite{blw}, if $(z_1,\cdots,z_r)^T$ is
a solution of $\sigma^n(Z)=\tilde{A}_{22}Z$,
then~$(0,\cdots,0,z_1,\cdots,z_r)^T$ is a solution of
$\sigma^n(Y)=\tilde{A}Y$. So $\sigma^n(Y)=\bar{A}_0Y$ has at least
$r$ solutions~$\overline{W}_1\bar{h}_1,
\dots,\overline{W}_r\bar{h}_r $, where~$\overline{W}_i \in \CX(x)^n$
and $\bar{h}_i$
satisfies~$\sigma^n(\bar{h}_i)=\alpha(x+k_{n-r+i})\bar{h}_i$. By
($A_2$), we can find all hypergeometric solutions of
$\sigma(Z)=\bar{A}_0(nx)Z$ where $\bar{A}_0(nx)$ means replacing $x$
by~$nx$ in $\bar{A}_0$. Then by interlacing, we can find all
solutions of $\sigma^n(Y)=\bar{A}_0Y$ of the
form~$\tilde{W}_j\tilde{h}_j$ where $\tilde{W}_j \in
\mathbb{C}(x)^n$ and $\tilde{h}_j$ satisfies
$\sigma^n(\tilde{h}_j)=\tilde{a}_j\tilde{h}_j$ for some~$\tilde{a}_j
\in \CX(x)$. Then there exists $\tilde{h}_{j_0}$ such that
$\tilde{h}_{j_0}=g\bar{h}_1$ for some~$g\in \CX(x)$ and
$$
\hat{\alpha}(x+k_{n-r+1})=\frac{\sigma^n(\tilde{h}_{j_0})}{\tilde{h}_{j_0}}=\frac{\sigma^n(g)}{g}\alpha(x+k_{n-r+1}).
$$
  After finding $\hat{\alpha}(x+k_{n-r+1})$, we can
compute a matrix $\hat{G} \in \mat_n(k_0)$ in a  finite number of steps by
($A_1$), such that
$$
\sigma^n(\hat{G}^{-1})A_n\hat{G}=\beta(t)\diag(\hat{\alpha}(x),\cdots,\hat{\alpha}(x+n-1)).
$$
Let $\bar{B}=\hat{G}^{-1}B\hat{G}-\hat{G}^{-1}\delta(\hat{G})$. Then
we get a new system
\begin{equation*}
   \sigma^n(Y)=\beta(t)\diag(\hat{\alpha}(x),\cdots,\hat{\alpha}(x+n-1))Y,
   \quad
   \delta(Y)=\bar{B}Y
\end{equation*}
which is equivalent to the original one under the transformation
$Y\rightarrow \hat{G}^{-1}Y.$ Since $\sigma^n$ and $\delta$ commute
and $\frac{\alpha(x+1)}{\alpha(x)}\neq \frac{\sigma^n(b)}{b}$ for
any $b \in \CX(x)$, the same argument as in the proof of Proposition
\ref{thm52} implies that $\bar{B}$ is of diagonal form, that is
$$
\bar{B}=
\diag\left(\frac{\delta(\beta(t))}{n\beta(t)}x+\hat{b}_1,\cdots,\frac{\delta(\beta(t))}{n\beta(t)}x+\hat{b}_n\right).
$$
 We then get the $\hat{b}_i$, and the $W_i$ are just the $i$-th columns of $\hat{G}$.
\begin{example} \label{example2}
Consider an
integrable system:
$$
\sigma(Y)=AY, \quad \delta(Y)=BY
$$
where
\begin{align*}
   &\begin{array}{ccc}
      A=
         \begin{pmatrix}
           \frac{x^3t^4+2x^2t^4+xt^4-x-1}{t^2+x+1}&
           \frac{t^2(tx^4+2tx^3+tx^2+1)}{t^2+x+1}&\frac{t(t-x-1)}{t^2+x+1}\\[0.1in]
         -\frac{t(x^2t^4+xt^4-1)}{t^2+x+1}&-\frac{t(t^3x^3+t^3x^2-1)}{t^2+x+1}&\frac{t(1+t)}{t^2+x+1}\\[0.1in]
         \frac{t^6x^2+t^6x+x+1}{t(t^2+x+1)}&\frac{t(t^3x^3+t^3x^2-1)}{t^2+x+1}&
         -\frac{t-x-1}{t^2+x+1}
         \end{pmatrix},
    \end{array} \\[0.1in]
   &\begin{array}{cc}
      B=\begin{pmatrix}
           \frac{t^4+t^2x+x^2+t^4x-t^2}{t(t^2+x)}
           &-\frac{x(-t^2+t^3-1)}{t^2+x}&\frac{xt^3(-1+t)}{t^2+x}\\[0.1in]
           -\frac{-t^2+t^4+1}{t^2+x}&
           \frac{2t^2x+x^2+t^5-t^2}{t(t^2+x)}&-\frac{t^4(-1+t)}{t^2+x}\\[0.1in]
        \frac{-t^2+t^4+1}{t^2+x}&\frac{x(-t^2+t^3-1)}{t(t^2+x)}&
        \frac{x^2+xt^3+t^2x+t^6-x-t^2}{t(t^2+x)}
      \end{pmatrix}.
   \end{array}
\end{align*}
We have
$$
  \det(A)=\frac{xt^3(t^2x+t^2+x^2+x)}{x+1+t^2}=\frac{(x+1)(t^2+x)}{x(t^2+x+1)}x^2t^3.
$$
By (\ref{eqn64}), if the Galois group over $k_0$ of the given system
is solvable by finite, then this system  have no hypergeometric
solutions over $k$. Therefore we consider the system
$$
   \sigma^3(Y)=A_3Y, \quad \delta(Y)=BY
$$
where
\[
\begin{array}{ccc}
      A_3=
         \begin{pmatrix}
           \frac{t^3(t^2x^2+t^2x+21x+x^3+8x^2+18)}{t^2+x+3}&
           -\frac{t^4(x+1)(5x+6)}{t^2+x+3}&\frac{2t^4(x+2)(x+3)}{t^2+x+3}\\[0.1in]
          -\frac{2t^4(2x+3)}{t^2+x+3}&\frac{(x+1)t^3(x^2+t^2x+2t^2)}{t^2+x+3}&-\frac{2t^5(x+2)}{t^2+x+3}\\[0.1in]
          \frac{2t^4(2x+3)}{t^2+x+3}&\frac{(x+1)t^3(5x+6)}{t^2+x+3}&
         \frac{(x+2)(x+3)t^3(x+1+t^2)}{t^2+x+3}
         \end{pmatrix}.
    \end{array} \]
We can compute $\beta(t)=t^3$ from $\det(A)$. Let
$\tilde{A}=\frac{A_3}{t^3}$. Then
\[
\begin{array}{ccc}
      \tilde{A}|_{t=0}=
         \begin{pmatrix}
           (x+2)(x+3)&0&0 \\
           0 & \frac{(x+1)x^2}{x+3}&0 \\
           0 & \frac{(x+1)(5x+6)}{x+3} & (x+1)(x+2)
         \end{pmatrix}.
    \end{array} \]
By $(A_2)$, all hypergeometric solutions of
$\sigma^3(Y)=\tilde{A}|_{t=0}Y$ are
\[
   \begin{array}{ccccc}
       9^{\frac{x}{3}}\Gamma\left(\frac{x+2}{3}\right)\Gamma\left(\frac{x+3}{3}\right)
         \begin{pmatrix}
           1 \\ 0\\0
         \end{pmatrix}
         ,& 9^{\frac{x}{3}}\Gamma\left(\frac{x+1}{3}\right)\Gamma\left(\frac{x+2}{3}\right)
          \begin{pmatrix}
             0\\0\\1
          \end{pmatrix},&
          9^{\frac{x}{3}}\Gamma\left(\frac{x}{3}\right)\Gamma\left(\frac{x+1}{3}\right)
          \begin{pmatrix}
             0 \\ -\frac{3}{x} \\ \frac{3}{x}
          \end{pmatrix},
   \end{array}
\]
where $\Gamma(x)$ satisfies $\Gamma(x+1)=x\Gamma(x).$ By $(A_1)$, we
can compute a rational solution of
$\sigma^3(Y)=\frac{A_3}{x(x+1)t^3}Y$. Moreover, we can compute a
matrix~$G \in \mat_3(\CX(x,t))$ such that
$$
   \sigma^3(G)\diag(x(x+1)t^3,(x+1)(x+2)t^3,(x+2)(x+3)t^3)=A_3G
$$
where
\[\begin{array}{ccc}
   G=
   \begin{pmatrix}
       \frac{t}{t^2+x} &  -\frac{x}{t^2+x} &  \frac{x}{t^2+x}\\
        \frac{1}{t^2+x} &  \frac{t}{t^2+x} &  -\frac{t}{t^2+x} \\
        -\frac{1}{t^2+x} &  \frac{x}{t(t^2+x)} &  \frac{t}{t^2+x}
   \end{pmatrix}.
\end{array}
\]
Let $\bar{B}=G^{-1}BG-G^{-1}\delta(G)$. Then $
   \bar{B}=\diag\left(\frac{x}{t}+t,\frac{x}{t}+t^2,\frac{x}{t}+t^3\right).
$ Hence a basis of solution space of
$\{\sigma^3(Y)=A_3Y,\delta(Y)=BY\}$ is
\[
   \begin{array}{ccc}
       V_1(x) := 9^{\frac{x}{3}}\Gamma\left(\frac{x}{3}\right)\Gamma\left(\frac{x+1}{3}\right)t^{x}e^{\frac{t^2}{2}}
         \begin{pmatrix}
           \frac{t}{t^2+x} \\ \frac{1}{t^2+x}\\-\frac{1}{t^2+x}
         \end{pmatrix},
  \end{array}
\]
   \[
   \begin{array}{ccc}
          V_2(x) := 9^{\frac{x}{3}}\Gamma\left(\frac{x+1}{3}\right)\Gamma\left(\frac{x+2}{3}\right)t^{x}e^{\frac{t^3}{3}}
          \begin{pmatrix}
             -\frac{x}{t^2+x} \\\frac{t}{t^2+x} \\ \frac{x}{t(t^2+x)}
          \end{pmatrix},
   \end{array}\]
\[
   \begin{array}{ccc}
   V_3(x):=9^{\frac{x}{3}}\Gamma\left(\frac{x+2}{3}\right)\Gamma\left(\frac{x+3}{3}\right)t^{x}e^{\frac{t^4}{4}}
          \begin{pmatrix}
             \frac{x}{t^2+x}\\-\frac{t}{t^2+x}\\ \frac{t}{t^2+x}
          \end{pmatrix}.
   \end{array}
\]
Clearly, $V_i(1) \neq 0$ for $i=1,2,3$, and~$A(j)$ and $B(j)$ are
well defined and~$\det(A(j))\neq 0$ for $j\geq 1$.
By the results in Section \ref{sec2}, we get a basis of the solution
space of the original system:
\[
   \begin{array}{ccc}
   W_1=9^{\frac{1}{3}}\Gamma(\frac{1}{3})\Gamma(\frac{2}{3})te^{\frac{t^2}{2}}
          \begin{pmatrix}
             (0,\frac{t}{t^2+1},\frac{4t^3}{t^2+2},-\frac{6t^3}{t^2+3},\cdots)\\[0.1in]
             (0,\frac{1}{t^2+1},-\frac{2t^4}{t^2+2},\frac{2t^4}{t^2+3},\cdots)\\[0.1in]
             (0,-\frac{1}{t^2+1},\frac{2t^4}{t^2+2},\frac{6t^2}{t^2+3},\cdots)
          \end{pmatrix},
   \end{array}
\]
\[
   \begin{array}{ccc}
   W_2=9^{\frac{1}{3}}\Gamma(\frac{2}{3})\Gamma(1)te^{\frac{t^3}{3}}
          \begin{pmatrix}
             (0,-\frac{1}{t^2+1},\frac{t}{t^2+2},\frac{18t^3}{t^2+3},\cdots)\\[0.1in]
             (0,\frac{t}{t^2+1},\frac{1}{t^2+2},-\frac{6t^4}{t^2+3},\cdots)\\[0.1in]
             (0,\frac{1}{t(t^2+1)},-\frac{1}{t^2+2},\frac{6t^4}{t^2+3},\cdots)
          \end{pmatrix}
   \end{array}
\]
and
\[
   \begin{array}{ccc}
   W_3=9^{\frac{1}{3}}\Gamma(1)\Gamma(\frac{4}{3})te^{\frac{t^4}{4}}
          \begin{pmatrix}
             (0,\frac{1}{t^2+1},-\frac{2}{t^2+2},\frac{t}{t^2+3},\cdots)\\[0.1in]
             (0,-\frac{t}{t^2+1},\frac{t}{t^2+2},\frac{1}{t^2+3},\cdots)\\[0.1in]
             (0,\frac{t}{t^2+1},\frac{2}{t(t^2+2)},-\frac{1}{t^2+3},\cdots)
          \end{pmatrix}.
   \end{array}
\]
{Note that all the~$W_i$ are  liouvillian. }
\end{example}

\subsection{Summary}
 Consider  two systems
\begin{equation}
   \sigma(Y)=AY, \quad\delta(Y)=BY \label{eqn65}
\end{equation}
and
\begin{equation}
   \sigma^n(Y)=A_nY, \quad \delta(Y)=BY \label{eqn66}
\end{equation}
where $A \in \mat_n(k_0)$, $B \in \gl_n(k_0)$ and $n$ is a prime
number. Assume that~(\ref{eqn65}) is irreducible over $k_0$. From
the results in Sections~\ref{sec6.1} and~\ref{sec3.2.2}, if
(\ref{eqn65}) has a liouvillian solution over $k$, then either the
solution space of (\ref{eqn65}) has a basis consisting of
hypergeometric solutions over $k$ or the solution space
of~(\ref{eqn66}) has a basis each of whose members is the
interlacing of hypergeometric vectors over $k_0$. Let us summarize
the previous decision procedure as follows.\\[0.2in]
 {\bf \underline{Decision Procedure 1}}
   Compute a fundamental matrix of (\ref{eqn65}) whose entries are
   hypergeometric over $k$ if it exists.
\begin{enumerate}
      \item [$(a)$]
      Write $\det(A)=\frac{\sigma(g)}{g}a$ where $g,a\in k_0$ and
$a$ is standard with respect to $\sigma$. If $a\neq
\alpha(x)^n\beta(t)$ for any $\alpha(x) \in \mathbb{C}(x)$ and~$
\beta(t)\in
       \mathbb{C}(t)$, then by the results in Section 6.1, {\bf exit} [(\ref{eqn65}) has no required fundamental
       matrix].
      \item [$(b)$]
        Assume that $a=\alpha(x)^n\beta(t)$ for some $\alpha(x) \in \mathbb{C}(x)$ and~$ \beta(t)\in
       \mathbb{C}(t)$.
        By the algorithms in \cite{barkatou-chen,barkatou}, compute
        an irreducible matrix $\tilde{A}$ such that
        $\tilde{A}=\sigma(\tilde{G})\frac{A}{\alpha(x)}\tilde{G}^{-1}$ for some $\tilde{G} \in
        \mat_n(k_0)$. If $\ord_{\infty}(\tilde{A})\neq
        0$, then by Proposition \ref{prop63}, 
        {\bf exit} [(\ref{eqn65}) has no required fundamental
       matrix].
         Otherwise, let $\tilde{A}_0=\tilde{A}|_{x=\infty}$ and~$\beta_1(t),\cdots,\beta_n(t)$ be the eigenvalues of~$\tilde{A}_0$.
       \item [$(c)$]
        {Goto
        Step ($d_1$) if the $\beta_i(t)$  are conjugate and goto Step
        ($d_2$) if $\beta_1(t)=\cdots =\beta_n(t) \in \mathbb{C}(t)$. }In other cases, by Lemma \ref{lem61}
        and Proposition~\ref{prop63},
        {\bf exit} [(\ref{eqn65}) has no required fundamental
       matrix].
       \item [$(d_1)$]
          If~$(A_1)$ yields no rational solutions, then  {\bf exit} [(\ref{eqn65}) has no required fundamental
       matrix]. Otherwise, suppose that we find~$G \in \mat_n(k)$ such that
         $$
         \sigma(G)\alpha(x)\diag(\beta_1(t),\cdots,\beta_n(t))=AG.
         $$
        Then $\bar{B}:=G^{-1}BG-G^{-1}\delta(G)$ 
       is of diagonal form. Compute a fundamental matrix~$H$
of
         $$
         \sigma(Y)=\alpha(x)\diag(\beta_1(t),\cdots,\beta_n(t))Y, \quad\delta(Y)=\bar{B}Y.
         $$
      {\bf Return} [$GH$ is a required fundamental
       matrix of~(\ref{eqn65})].
\item [$(d_2)$]
         If we can compute a matrix~$G \in \mat_n(k_0)$ such that
         $\sigma(G)\alpha(x)\beta_1(t)=AG$ then
         let
         $$
         \hat{B}=G^{-1}BG-G^{-1}\delta(G)-\frac{\delta(\beta_1(t))}{\beta_1(t)}xI_n \in \gl_n(\mathbb{C}(t)),
         $$
         else  {\bf exit} [(\ref{eqn65}) has no required fundamental
       matrix]. 
         If we can find a
         fundamental matrix $H$ of $\delta(Y)=\hat{B}Y$ whose
         entries are hyperexponential over $\overline{\mathbb{C}(t)}$,
then {\bf reurn} [$GHh\beta_1(t)^x$  is a required fundamental
matrix of (\ref{eqn65})] where $h$ satisfies $\sigma(h)=\alpha(x)h$
and $\delta(h)=0$. 
Otherwise, 
 {\bf exit} [(\ref{eqn65}) has no required fundamental
       matrix].
  \end{enumerate}
  {\bf \underline{Decision Procedure 2}}
   Compute a fundamental matrix of (\ref{eqn66}) whose entries are
   the interlacing of hypergeometric vectors over $k_0$ if it exists.
  \begin{enumerate}
     \item [$(a)$]
        If $\det(A)\neq (-1)^{n-1}\frac{\sigma(g)}{g}\alpha(x)\beta(t)$ holds for any $g \in
        k$,~$ \beta(t)\in
       \mathbb{C}(t)$ and~$\alpha(x) \in \mathbb{C}(x)$ that is standard with respect to
       $\sigma^n$,
       then {\bf exit} [(\ref{eqn66}) has no required fundamental
          matrix].
     \item [$(b)$]
         Expand $\det(A)$ as a series at $x=\infty:$
         $$
         \det(A)=(-1)^{n-1}\beta(t)x^m+\beta_1(t)x^{m-1}+\cdots
         $$
         where $\beta(t),\beta_i(t) \in \CX(t)$ and $m\in \bZ$.
         Suppose that $x=p$ is not a pole of the entries of
         $\frac{A}{\beta(t)}$ and that $\det(\tilde{A}_0) \neq 0$ where $\tilde{A}_0=\frac{A}{\beta(t)}|_{x=p}$.
         Use~($A_2$) to find all hypergeometric solutions of $\sigma(Z)=\tilde{A}_0(nx)Z$.
         By interlacing,
         we get all solutions of $\sigma^n(Y)=\tilde{A}_0Y$
        of the form $W_ih_i$. Denote these soluions by
         $W_1h_1,\cdots,W_dh_d$ where  $W_i \in \mathbb{C}(x)^n$ and  $h_i$
         satisfies $\sigma^n(h_i)=\tilde{a}_ih_i$ for some $\tilde{a}_i \in
         \CX(x)$. If there is~$i_0\in\{1, \dots,  d\}$ such
         that $\sigma^n(Y)=\frac{h_{i_0}A}{\sigma^n(h_{i_0})\beta(t)}Y$ has a
         rational solution in $k_0^n$, then let~$\lambda(x)=\frac{\sigma^n(h_{i_0})}{h_{i_0}}$, else {\bf exit} [(\ref{eqn66})
has no required fundamental
          matrix]. Let $j_0$
         be the least integer such that
         $\sigma^n(Y)=\frac{A}{\lambda(x+j_0)\beta(t)}Y$ has a
         rational solution in $k_0^n$. If we can compute  $G \in
         \mat_n(k_0)$ such that
         $$
         \sigma(Y)\beta(t)\diag(\lambda(x+j_0),\cdots,\lambda(x+j_0+n-1))=AG,
         $$
         then let $\bar{B}=G^{-1}BG-G^{-1}\delta(G)$. So $\bar{B}$ is of diagonal form  and
         by the same process as in Step ($d_1$) of Decision Procedure 1, we can
         compute a required fundamental matrix of (\ref{eqn66}).
         Otherwise, by the results in Section \ref{sec3.2.2},
         {\bf exit} [(\ref{eqn66}) has no required fundamental
          matrix].
\end{enumerate}
We can decide whether (\ref{eqn65}) has  liouvillian solutions or
not as follows. If we can compute hypergeometric solutions over $k$
of (\ref{eqn65}) by Decision Procedure~1, then we are done.
Otherwise, consider the system (\ref{eqn66}). If we can compute
liouvillian solutions over $k_0$ of (\ref{eqn66}) by Decision
Procedure 2, then by the results in Section \ref{sec2} we can
compute liouvillian solutions over $k_0$ of~(\ref{eqn65}) and we are
done. Otherwise~(\ref{eqn65}) has no liouvillian solutions.
\\[0.1in]
%

\end{document}